\documentclass[10pt,twocolumn,twoside]{IEEEtran}

\usepackage{acronym}

\newcommand{\avec}{{\bf{a}}}

\newcommand{\cvec}{{\bf{c}}}

\newcommand{\evec}{{\bf{e}}}

\newcommand{\Yvec}{{\bf{Y}}}
\newcommand{\yvec}{{\bf{y}}}
\newcommand{\uvec}{{\bf{u}}}

\newcommand{\xvec}{{\bf{x}}}
\newcommand{\zvec}{{\bf{z}}}

\newcommand{\svec}{{\bf{s}}}
\newcommand{\vvec}{{\bf{v}}}

\newcommand{\onevec}{{\bf{1}}}
\newcommand{\zerovec}{{\bf{0}}}

\newcommand{\nuvec}{{\boldsymbol{\nu}}}

\newcommand{\Bmat}{{\bf{B}}}

\newcommand{\Hmat}{{\bf{H}}}

\newcommand{\Imat}{{\bf{I}}}
\newcommand{\Lmat}{{\bf{L}}}

\newcommand{\Rmat}{{\bf{R}}}
\newcommand{\Umat}{{\bf{U}}}

\newcommand{\Ymat}{{\bf{Y}}}

\newcommand{\define}{\stackrel{\triangle}{=}}

\def\bLambda{{\mbox{\boldmath $\Lambda$}}}

\def\thetavec{{\mbox{\boldmath $\theta$}}}

\newcommand{\be}{\begin{equation}}
\newcommand{\ee}{\end{equation}}
\newcommand{\beqna}{\begin{eqnarray}}
\newcommand{\eeqna}{\end{eqnarray}}

\acrodef{gsp}[GSP]{graph signal processing}
\acrodef{dsp}[DSP]{digital signal processing}
\acrodef{gnn}[GNN]{graph neural network}
\acrodef{glrt}[GLRT]{generalized likelihood ratio test}
\acrodef{gso}[GSO]{graph shift operator}
\acrodef{gft}[GFT]{Graph Fourier Transform}
\acrodef{igft}[IGFT]{inverse GFT}
\acrodef{gsp}[GSP]{graph signal processing}
\acrodef{ghpf}[GHPF]{graph high pass filter}
\acrodef{gfdi}[GFDI]{graph false data injection}
\acrodef{pgd}[PGD]{projected gradient descent}
\acrodef{bdd}[BDD]{bad data detection}
\acrodef{admm}[ADMM]{alternating direction method of multipliers}
\acrodef{psse}[PSSE]{Power system state estimation}
\acrodef{fdi}[FDI]{false data injection}
\acrodef{tv}[TV]{total variation}
\acrodef{dc}[DC-PF]{direct current power flow}
\acrodef{ac}[AC-PF]{alternating current power flow}
\acrodef{pmu}[PMU]{phasor measurement unit}
\acrodef{roc}[ROC]{receiver operating characteristics}
\acrodef{ems}[EMSs]{energy management systems}
\acrodef{dc}[DC]{Direct Current}
\acrodef{ac}[AC]{Alternating Current}
\acrodef{glpf}[GLPF]{graph low pass filter}

\newcommand{\Gcal}{\mathcal{G}}
\newcommand{\Vcal}{\mathcal{V}}

\newcommand{\Scal}{\mathcal{S}}

\newcommand{\Dcal}{\mathcal{D}}

\newcommand{\Rupm}{{\mathbb{R}^{M}}}

\newcommand{\upperRomannumeral}[1]{\uppercase\expandafter{\romannumeral#1}}

\newcommand{\Hnull} {\mathcal{H}_{0}}
\newcommand{\Ha} {\mathcal{H}_{1}}

\newcommand{\hbt}{\hat{\boldsymbol{\theta}}}

\usepackage{amssymb, amsmath,amsthm}
\usepackage[dvips]{color} 
\usepackage{lscape,longtable,nomencl}
\usepackage[usenames,dvipsnames]{pstricks}
\usepackage{verbatim}
\usepackage{mathtools}
\usepackage{graphicx}
\usepackage{caption}
\usepackage{subcaption}

\usepackage[labelfont=bf]{caption}
\usepackage{stackrel}
\usepackage{setspace}
\usepackage[resetlabels,labeled]{multibib}
\usepackage{enumitem,lipsum}
\usepackage{cite}
\usepackage{array}

\graphicspath{{figures/}} 

\usepackage{mathtools}
\usepackage{boldline,multirow}
\usepackage{color,soul}

\usepackage{bbm}

\usepackage{times} 

\usepackage[utf8]{inputenc}
\usepackage[english]{babel}

\usepackage{lipsum}

\newcommand{\RNum}[1]{\uppercase\expandafter{\romannumeral #1\relax}}

\usepackage{mathtools}

\usepackage{tikz}
\usetikzlibrary{shapes,arrows}

\newcommand{\tlvert}{\lvert\lvert}
\newcommand{\trvert}{\rvert\rvert}

\newtheoremstyle{ThDef}
{}                
{}                
{\slshape}        
{}                
{\bfseries}       
{.}               
{0.15cm }               
{}                

\theoremstyle{ThDef}

\newtheorem{thm}{Theorem}

\DeclareFontEncoding{LS1}{}{}
\DeclareFontSubstitution{LS1}{stix}{m}{n}
\DeclareSymbolFont{symbols2}{LS1}{stixfrak} {m} {n}
\DeclareMathSymbol{\operp}{\mathbin}{symbols2}{"A8}

\usepackage{algorithm,algorithmic} 
\begin{document}
	
	\captionsetup[figure]{labelfont={bf},labelformat={default},labelsep=period,name={Fig.}}

\title{Protection Against Graph-Based False Data Injection Attacks on Power Systems}
\author{Gal Morgenstern,~\IEEEmembership{Student Member,~IEEE}, 
Jip Kim, ~\IEEEmembership{Member,~IEEE},
James Anderson,~\IEEEmembership{Senior Member,~IEEE},\\
Gil Zussman,~\IEEEmembership{Fellow,~IEEE},
Tirza~Routtenberg,~ \IEEEmembership{Senior Member,~IEEE}
\thanks{G.\ Morgenstern and T.\ Routtenberg are with the School of Electrical and Computer Engineering,   Ben-Gurion University of the Negev,
Beer-Sheva 84105, Israel, Email: \{galmo,tirzar\}@post.bgu.ac.il. T.\ Routtenberg is also with the Department of Electrical and Computer Engineering, Princeton University, Princeton, NJ. J.\ Kim is with KENTECH, South Korea, Email: jipkim@kentech.ac.kr. J.\ Anderson and G.\ Zussman are with the Department of Electrical Engineering, Columbia University, New York, NY, Email: \{james.anderson,gil.zussman\}@columbia.edu.\\
© 20XX IEEE.  Personal use of this material is permitted.  Permission from IEEE must be obtained for all other uses, in any current or future media, including reprinting/republishing this material for advertising or promotional purposes, creating new collective works, for resale or redistribution to servers or lists, or reuse of any copyrighted component of this work in other works.	}}

	\maketitle

\begin{abstract}
 Graph signal processing (GSP) has emerged as a powerful tool for practical network applications, including power system monitoring. 
Recent research has focused on developing GSP-based methods for state estimation, attack detection, and topology identification using the representation of the power system voltages as smooth graph signals.
Within this framework, efficient methods have been developed for detecting false data injection (FDI) attacks, which until now were perceived as non-smooth with respect to the graph Laplacian matrix. 
 Consequently, these methods may not be effective against smooth FDI attacks. 
 In this paper, we 
propose a graph FDI (GFDI) attack that minimizes the Laplacian-based graph total variation (TV) under practical constraints.
We present the GFDI attack as the solution for a non-convex constrained optimization problem. 
The solution to the GFDI attack problem is obtained through approximating it using $\ell_1$ relaxation.  A series of quadratic programming problems that are classified as convex optimization problems are solved to obtain the final solution.  
We then propose a protection scheme that identifies the minimal set of measurements necessary to constrain the GFDI output to a high graph TV, thereby enabling its detection by existing GSP-based detectors.  Our numerical simulations on the IEEE-$57$  and IEEE-$118$  bus test cases reveal the potential threat posed by well-designed GSP-based FDI attacks. Moreover, we demonstrate that integrating the proposed protection design with GSP-based detection can lead to significant hardware cost savings compared to previous designs of protection methods against FDI attacks.
	\end{abstract}

	\vspace{-0.3cm}
	\section{Introduction}
 
\ac{psse} is a crucial component of modern \ac{ems} that fulfills various purposes, including monitoring, analysis, security, control, and management of power delivery \cite{abur2004power}. 
 \ac{psse} is conducted using power measurements to estimate the voltages (states) at the system buses.
To ensure the reliability of the measurements, residual-based  \ac{bdd}
methods are integrated into the EMS \cite{abur2004power}. 
However, \ac{bdd} methods are not able to detect well-designed attacks, known as unobservable \ac{fdi} attacks \cite{liu2011false,liang2017review}, 
which can cause significant damage by misleading the \ac{psse} system\cite{xie2011integrity,jia2014impact,soltan2015analysis}.
These attacks are achieved by manipulating measurements based on the power network topology \cite{liu2011false},
where the topology matrix is either known or can be estimated from historical data \cite{kekatos2015online,kim2014subspace,grotas2019power,Halihal_Routtenberg_2022}.

Defending power systems against unobservable \ac{fdi} attacks involves two primary approaches. The first approach is to prevent attacks by protecting a subset of measurements using techniques such as encryption, continuous monitoring, and separation from the Internet \cite{kim2011strategic}.
This often involves identifying a minimal set of measurements required to prevent an adversary from constructing a feasible sparse \ac{fdi} attack \cite{bi2014graphical, deng2015defending, kim2011strategic, ansari2018graph,sou2019protection}. These works aim to ensure network observability and maintain the grid's immunity to well-coordinated attacks. Synchronized \ac{pmu} placement has also been suggested for the optimal deployment of protective measurements \cite{kim2011strategic, chen2006placement, kim2013phasor, sun2023asymptotic}. However, current methodologies do not consider recent developments of 
\ac{gsp}-based detectors against \ac{fdi}. 
The second main approach to protect power systems is to develop detection 
 methods against unobservable \ac{fdi} attacks that rely on system characteristics. 
 These methods include  
compressive sensing algorithms \cite{liu2014detecting, hao2015sparse, gao2016identification, morgenstern2022structural} that require certain structural properties for the system powers and a differential model with multi-time measurements. 
Another detection technique is the moving target defense \cite{ghaderi2020blended}, where the system configuration is actively changed.
Detection and identification methods based on machine learning, Kalman filters, and data mining have also been suggested \cite{esmalifalak2017detecting, wang2017novel, he2017real, almutairy2021accurate, wang2019online, kim2022identification, 6897944, chattopadhyay2019security}. 
However, these data-driven methods require a large set of historical and real-time power system data, which is usually unavailable. 

Designing unobservable \ac{fdi} attacks has also been investigated in the literature (see, e.g., \cite{liang2017review} and references therein). 
Several studies,  such as \cite{liu2011false, kim2011strategic, bobba2010detecting}, have examined the generation of valid unobservable FDI attacks with constraints on the adversary resources and access to the system sensors. 
Other researchers have focused on generating unobservable \ac{fdi} attacks where the adversary has incomplete knowledge of the power grid topology \cite{rahman2012false, liu2015modeling}. 
In some cases, designing an unobservable \ac{fdi} attack involves manipulating discrete data to reflect a false system topology \cite{kim2013topology}. 
Data-driven techniques such as partial component analysis (PCA) \cite{kim2014subspace, yu2015blind}, random matrix theory \cite{lakshminarayana2020data, kim2014subspace}, and learning \cite{tian2018data} have been used as well. 
However, these approaches do not consider the graphical representation of the power system, and, thus, may not fully leverage the benefits of \ac{gsp} techniques in detecting and mitigating \ac{fdi} attacks. 
Thus, incorporating graph-based techniques into the design of \ac{fdi} attack detection and mitigation methods can significantly enhance the resilience of power systems against cyber attacks.

GSP is a new and emerging field that extends concepts and techniques from traditional \ac{dsp} to data on graphs. 
GSP theory includes methods such as the \ac{gft}, graph filters \cite{sandryhaila2013discrete,ortega2018graph,shuman2013emerging}, and sampling and recovery of graph signals \cite{chen2015discrete,marques2015sampling,kroizer2022routtenberg}.
In recent years, tools from graph theory and \ac{gsp} have shown promise in the design of cyber-attack detection methods for power systems 
\hspace{-0.1cm} \cite{soltan2016power, bi2014graphical, 8372466,drayer2019detection,ramakrishna2021grid,Dabush_SAM_conf,dabush2023state,shereen2022detection}.
Specifically, theoretical analysis and experimental studies
have demonstrated 
that the system state vectors are low-pass graph signals \cite{drayer2019detection,ramakrishna2021grid, dabush2023state}. 
The works in \cite{drayer2019detection,ramakrishna2021grid,dabush2023state,shereen2022detection} leveraged this property to design  \ac{gsp}-based detectors that are able to detect unobservable \ac{fdi} attacks.
Since then, GSP and graph neural network (GNN) approaches have been widely adopted and applied in a variety of scenarios for FDI attack detection  (see, e.g., \cite{boyaci2021graph,jorjani2020graph,hasnat2022graph,haghshenas2023temporal,boyaci2022cyberattack,liu2023distributed}).
 The analysis and the detection design have been formulated for both \ac{dc} and \ac{ac} models, and with various types of measurements, such as \ac{pmu} data. 
However, the vulnerability of existing \ac{gsp}-based detection methods to  graph low-pass attacks has not been investigated,  but has only been mentioned as a topic for future research   \cite{shereen2022detection}.
Moreover, there is no practical method for generating an \ac{fdi} attack that exploits the graphical properties of the states.

\vspace{-0.2cm}
In this paper, we investigate the resilience of  \ac{psse} against \ac{fdi} attacks by leveraging the low graph \ac{tv} property of power system state variables. First, we introduce a novel type of unobservable \ac{fdi} attack, called the \ac{gfdi} attack, which is specifically designed to bypass \ac{gsp}-based detectors. Then, we propose a low-complexity solution to the non-convex \ac{gfdi} attack optimization problem,  which involves quadratic programming \cite{boyd2004convex}. 
Moreover, we propose a protection scheme that identifies the minimal set of secured sensors needed to prevent the damage of the unobservable \ac{gfdi} attack. We then present a practical greedy algorithm for the implementation of this scheme. Our simulation results demonstrate the vulnerability of existing \ac{gsp}-based detectors to the proposed \ac{gfdi} attack design.
Moreover, the simulations reveal that the proposed protection scheme significantly increases the graph \ac{tv} of the \ac{gfdi} attack, even when securing only a small portion of measurements, and thus makes the \ac{gfdi} attack detectable by the existing \ac{gsp}-based detectors.
Hence, the proposed protection scheme, when combined with a \ac{gsp}-detector, provides a cost-effective hybrid defense layer against \ac{fdi} attacks.

The remainder of this paper is organized as follows. 
In Section \ref{sec; background}, we introduce the necessary background on \ac{gsp}-based detection, \ac{psse}, and unobservable \ac{fdi} attacks. 
The \ac{gfdi} attack is introduced in Section \ref{sec; gfdi}. This attack is then used in Section \ref{sec; protection} to develop the \ac{gsp}-based protection scheme. Next, a simulation study is presented in Section \ref{sec; simulations}, and the conclusions appear in Section \ref{sec; conclusions}. 
	 
In this paper, vectors are denoted by boldface lowercase letters
and matrices are denoted by boldface uppercase letters.
The operators $\lvert\lvert \cdot \rvert\rvert$, $\lvert\lvert \cdot \rvert\rvert_0$,
and $\lvert\lvert \cdot \rvert\rvert_{\infty}$ denote the Euclidean norm, the zero semi-norm, and the max-norm, respectively. 
The operators  $(\cdot)^T$ and $(\cdot)^{-1}$ are the transpose and inverse operators, respectively. 
The notation $\Hmat^{\Scal}$ is the submatrix formed by the $\lvert \Scal \rvert$ rows of $\Hmat$ indicated by the indices in $\Scal$, where $\lvert \cdot \rvert$ is the cardinality of the input set. 
Finally, ${\text{diag}}(\avec)$ is the diagonal matrix whose $n$th diagonal entry is $a_n$.
 		
 \section{GSP-based \ac{fdi} attack detection }\label{sec; background}
 
 \subsection{GSP background} \label{sec; GSP pre}

We consider connected, undirected graphs, $\Gcal(\Vcal,\xi)$,   defined by a set of $N$ nodes, $\Vcal$, labeled $1,\hdots, N$, and a set of edges, $\xi$. 
Associated with each edge $(k,l)\in \Vcal \times \Vcal$ is a nonnegative weight denoted by $\omega_{k,l}$, 
unless there is no edge between nodes $k$ and $l$, and then, $\omega_{k,l}=0$. 
The graph nodes represent the entities of interest (e.g., users, items, sensors, etc.), and the edges define the interactions between them. 
These interactions are captured by the graph Laplacian matrix,
\be \label{eq; L}
\Lmat_{k,l}=\begin{cases}
\sum_{m\in\mathcal{N}_k} \omega_{k,m} & k=l \\
-\omega_{k,l} &  ~(k,l)\in \xi \\
0 & \text{otherwise},
\end{cases}
\ee
where $\mathcal{N}_k$ is the set of buses connected to bus $k$.

    	Given a graph $\Gcal(\Vcal,\xi)$, a graph signal is defined by the mapping:
		\be \label{eq; signal} 
		\svec:~\Vcal\rightarrow \mathbb{R}^N,
		\ee
		where each coordinate of the state variable in $\svec$ is assigned to one of the system nodes,
		i.e.,  $s_n$ denotes the signal value at node $n$.
	    Equivalently to in \ac{dsp} literature, the graph signal can also be represented in the (graph) spectral domain. 
		This representation is provided by the following
		 graph Laplacian eigendecomposition:
		\be \label{eq; eigen decomposition}
		\Lmat = \Umat \bLambda \Umat^T.
		\ee
		In \eqref{eq; eigen decomposition}, the diagonal matrix, $\bLambda={\text{diag}}(\lambda_1,\ldots,\lambda_n)$, contains the eigenvalues of 
  $\Lmat$, which are referred to as the graph frequencies in the \ac{gsp} literature.
 These eigenvalues are real and are assumed to be ordered as
		\be \label{eq; eigen values}
		0=\lambda_1<\lambda_2\le \ldots \le \lambda_n,
		\ee
  where the strict inequality between $\lambda_1$ and $\lambda_2$ follows because we only consider connected graphs.
		In addition, $\Umat$ is a matrix whose $n$th column, $\uvec_n$, is the eigenvector of $\Lmat$ associated with $\lambda_n$, and $\Umat^{T}=\Umat^{-1}$.
 The \ac{gft} of the graph signal $\svec$ is 
		\be \label{eq; GFT}
		\tilde{\svec}\define \Umat^T\svec.
		\ee 
		The resulting signal, $\tilde{\svec}$, has properties analogous to the discrete Fourier transform of time series \cite{ortega2018graph,zhu2012approximating}. 
       The \ac{igft} is given by
		$\svec \define \Umat\tilde{\svec}$.

		Analogous to classical \ac{dsp} theory, a graph filter is a system with a graph signal as an input and another graph signal as an output.
	 The filtering process is often defined by the filter frequency response $f(\cdot)$ as \cite{ortega2018graph}
		\be \label{eq; graph filter decomposition}
		f(\Lmat)=\Umat  f(\bLambda)  \Umat^T,
		\ee
		where $f(\bLambda)={\text{diag}}(f(\lambda_1),\ldots,f(\lambda_n))$, and $f(\lambda_n)$ is the graph  filter  frequency response  at the graph frequency $\lambda_n$.  
\begin{figure*}[hbt]
\centering
\includegraphics[trim={0.5cm 1.8cm 0.4cm 2.5cm},clip,width=14cm]{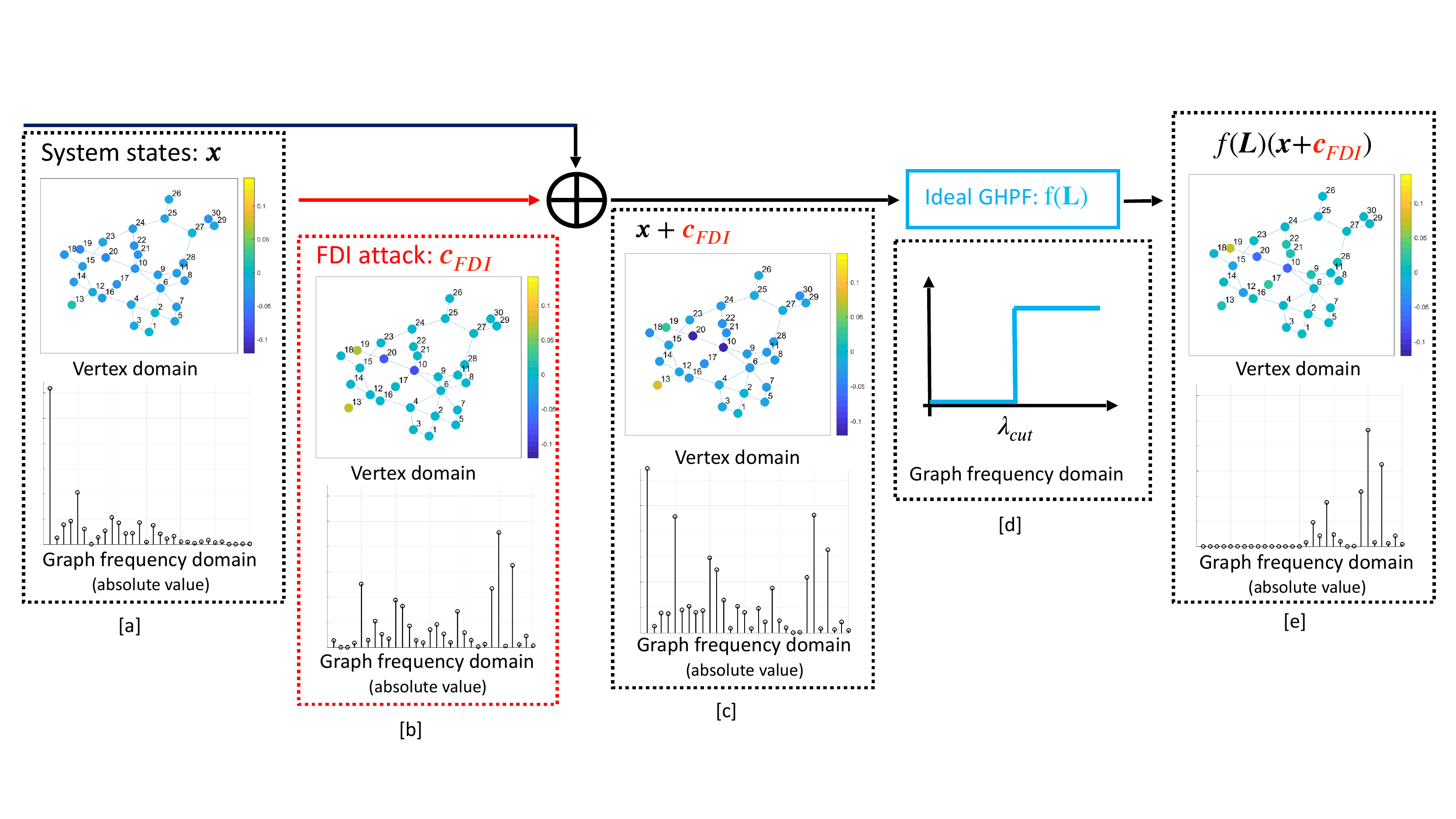} 
\caption{
    Illustration of GSP-based detection of \ac{fdi} attacks for the IEEE-30 bus system:
    the system states are presented in block [a] in both the vertex and graph frequency domains, where it can be seen that the states are smooth, low-frequency graph signals.
    In block [b], the state attack ($\cvec$ from \eqref{eq; FDI attack}) is presented, and block [c] presents the contaminated states. 
    It can be seen that, in contrast to the states, the state attack in [b] and the contaminated states in [c] 
    are not smooth graph signals and they have significant energy at the higher graph frequencies. 
    After filtering the contaminated states by the GHPF  defined in \eqref{eq; GHPF ideal} (shown in block [d]),
    the filtered signal in block [e] contains only energy at the higher graph frequencies. 
    Consequently, the attack can be discovered by the detector in \eqref{eq; detection GFT2}. }
\label{fig; FDI schematic}
\end{figure*}
	
\subsection{GSP smoothness measures} \label{sec; GSP anomaly}
The graph \ac{tv} of a graph signal $\svec$ is defined as \cite{shuman2013emerging} 
\be \label{eq; TV}
\begin{aligned}
    TV^{\Gcal}(\svec)\define \svec^T\Lmat\svec =\frac{1}{2} \sum _{k=1}^N\sum_{n=1}^N \omega_{k,n}\big(s_k - s_n\big)^2,
\end{aligned}
\ee
where the last equality is obtained by substituting \eqref{eq; L}.
The \ac{tv} in \eqref{eq; TV} is a smoothness measure, which is used to quantify changes w.r.t. the variability that is encoded by the weights of the graph 
\cite{ortega2018graph}. 
By substituting \eqref{eq; eigen decomposition} and \eqref{eq; GFT} in  \eqref{eq; TV}, we obtain 
\be \label{eq; TV freq}
\begin{aligned}
TV^{\Gcal}(\svec)=\sum_{k=1}^N\lambda_k\tilde{s}_k^2.
\end{aligned}
\ee
According to \eqref{eq; TV freq}, if $\svec$ is smooth, its \ac{gft} representation from \eqref{eq; GFT}, $\tilde{\svec}$, decreases as the graph frequency increases.
Thus, low graph \ac{tv} forces the graph spectrum of the signal to be concentrated in the small eigenvalues region.  

The concept of smoothness w.r.t. the graph has been generalized in \ac{gsp} theory. In \ac{gsp} theory, the smoothness of the graph signal $\svec$ is given by 
\be \label{eq; detection GFT}
T^f(\svec)= 	\tlvert f(\Lmat)\svec\trvert^2 
,
\ee
where $f(\Lmat)$ is a \ac{ghpf}.
The \ac{ghpf} is a graph filter as defined in \eqref{eq; graph filter decomposition}, where its frequency response, $f(\lambda_i)$, maintains lower values at the higher graph frequencies. 
In particular, using 
 the graph frequency response
\be \label{eq; GHPF TV} 
f^{TV}(\lambda_i) =\sqrt{\lambda_i},  ~i = 1,\ldots,N,
\ee
 we obtain that the smoothness measure  in \eqref{eq; detection GFT} for this case is reduced to the graph \ac{tv} from \eqref{eq; TV freq}.
Another example is the ideal \ac{ghpf}, which is defined by the frequency response \cite{Sandryhaila_Moura_2014}
\be \label{eq; GHPF ideal}
f^{id}(\lambda_i) = \begin{cases}
    0 & \lambda_i \le \lambda_{cut} \\
    1 & \lambda_i>\lambda_{cut} 
\end{cases}~~~i = 1, \ldots,N,
\ee
where the cutoff frequency, $\lambda_{cut}$, can be determined based on the application.

	   \subsection{Unobservable FDI attacks}
	A power system 
	is a network of buses (generators or loads) connected by transmission lines that can be represented as an undirected weighted graph, ${\mathcal{G}}({\mathcal{V}},\xi)$, where the set of nodes, $\mathcal{V}$, is the set of  $N$ buses, and the edge set, $\xi$, is the set of $P$ transmission lines between these buses.
	We denote the set of all
	sensor measurements by $\mathcal{M}$, and the set of transmission lines by $\xi$. 
        We consider the \ac{dc} model, in which each transmission line, $(k,n)\in \xi$, that connects buses $k$ and $n$ is characterized by a 
        susceptance value $b_{k,n}$,
       where the branch resistances are neglected 
      \cite{abur2004power}. 
        The following noisy and attacked measurement model is considered \cite{abur2004power}:
	\begin{equation} \label{eq; DC model}
		\zvec =\Hmat\thetavec+\avec+\nuvec,
	\end{equation}
	where $\zvec = [z_1,\dots,z_M]^T \in {\mathbb{R}}^M$ is the vector of active powers,
        $\thetavec=[\theta_1,\dots,\theta_N]^T\in {\mathbb{R}}^N$ are the state variables (voltage phases), and 
         $\Hmat\in {\mathbb{R}}^{M\times N}$ is the measurements matrix,  defined as follows.  
         If row $r$ is associated with the power flow measurement on line $(k,n)$, then 
         \be \label{eq; H flows}
         H_{\{r,j\}}=\begin{cases} 
         b_{k,n} & j=k \\
         -b_{k,n} & j=n \\
         0 & \text{otherwise}.
         \end{cases}
         \ee
         Otherwise, if row $r$ is associated with the power injection measurement in substation $k$, then 
         \be \label{eq; H injections}
         H_{\{r,j\}}=\begin{cases} 
         \sum_{n\in\mathcal{N}_k} b_{k,n} & j=k \\
          -b_{k,j} & j\in \mathcal{N}_k \\
         0 & \text{otherwise},
         \end{cases} 
         \ee
         where $\mathcal{N}_k$ is defined in \eqref{eq; L}. 
         Hence, the measurement matrix  is determined by the topology of the network, the susceptance of the transmission lines, and the meter locations.
         In addition, 
         $\avec\in\Rupm$ models an \ac{fdi} attack
         and $\nuvec\in {\mathbb{R}}^M$ is the measurement noise modeled as a 
	zero-mean Gaussian vector with covariance  $\Rmat$.
     
 In the considered setting, any subset of measurements can be regarded as part of the set encompassing active power injections and power flows, as described by the model in \eqref{eq; DC model} with the measurement matrices in  \eqref{eq; H flows} and \eqref{eq; H injections}. For the sake of simplicity and to ensure that the attack on the states will have an impact on the \ac{psse} approach, in this paper we assume observability of the system. That is, it is assumed that the set of measurements is such that all state variables can be estimated from the available measurements through standard \ac{psse}.

An \ac{fdi} attack is defined to be {\em{unobservable}} if 
\be \label{eq; FDI attack}
\avec=\Hmat\cvec, 
\ee
where the state attack $\cvec\in{\mathbb{R}}^N$ is a nonzero arbitrary vector. 
The attack in \eqref{eq; FDI attack}
cannot be detected by classical residual-based \ac{bdd} methods  \cite{liu2011false}.
This can be seen  by substituting \eqref{eq; FDI attack} into \eqref{eq; DC model}, which results in
\be \label{eq; DC model c}
\zvec =\Hmat(\thetavec+\cvec)+\nuvec.
\ee 
Thus, the residual calculated with assaulted measurements, $\zvec-\Hmat(\hat{\thetavec}+\cvec)$, is the same as it is for normal measurements, $\zvec_{normal} =\Hmat\thetavec+\nuvec$ (see, e.g., \cite{liu2011false,liang2017review}). 
At the same time, these attacks can be designed to have severe physical \cite{liang2015vulnerability,liang2014cyber} and economic consequences \cite{8219710,jia2012impacts}.

\vspace{-0.4cm}
   \subsection{GSP-based \ac{fdi} attack detection}
\label{GSP_FDI_detection}

The problem of detecting \ac{fdi} attacks based on the  \ac{dc} model in \eqref{eq; DC model} can be formulated as 
 the following hypothesis-testing problem:
\be \label{eq; hypothesis power}
\begin{cases}
\Hnull:~\hat{\thetavec}=\thetavec+\bar{\nuvec}  \\
\Ha:~\hat{\thetavec}=\thetavec+\cvec +\bar{\nuvec},
\end{cases}
\ee
 where $\hat{\thetavec}$ is the \ac{psse} output based on $\zvec$, and $\bar{\nuvec}$ is the error or noise term associated with the estimation.
 Recently, it has been shown that GSP-based detectors in the form of
\cite{drayer2019detection,ramakrishna2021grid,dabush2023state,Dabush_SAM_conf,shereen2022detection}: 
 \be \label{eq; detection GFT2}
T^f(\hat{\thetavec})\mathop{\gtrless}_{\Hnull}^{\Ha} \gamma,
\ee
where $T^f(\cdot)$ is defined in 
 \eqref{eq; detection GFT},
are able to solve the hypothesis-testing problem in
\eqref{eq; hypothesis power}.
 In this case, the Laplacian matrix of the graph, $\Lmat$, is selected to be
 the nodal admittance matrix, $\Bmat$,  
 which is a submatrix of $\Hmat$, 
  composed of the rows in $\Hmat$ associated with the power injection measurements described in \eqref{eq; H injections}. 
This selection is applied by setting the edge weights of the graph as $\omega_{k,l}=-b_{k,l}$, where $b_{k,n}<0$ is the susceptance of line $(k,n)\in \xi$. 
Thus, by 
 by substituting $\omega_{k,l}=-b_{k,l}$
 in \eqref{eq; L} we obtain the nodal admittance matrix, $\Bmat$. 

 The \ac{gsp}-based detector in \eqref{eq; detection GFT2} is based on the assumption that the state vector, $\thetavec$, is a smooth graph signal w.r.t. to $\Lmat=\Bmat$, i.e., that $TV^{G}(\thetavec)= \thetavec^T\Lmat\thetavec$
 is small compared to other signals in the system,
 as shown in 
 \cite{dabush2023state,ramakrishna2021grid}.
 In contrast, the state attack vector, 
 $\cvec$, is a general arbitrary vector that is not smooth w.r.t. the graph.
 The \ac{gsp}-based detector in  \eqref{eq; detection GFT2}  was implemented in \cite{drayer2019detection} 
with the ideal GHPF defined in \eqref{eq; GHPF ideal}, and in \cite{dabush2023state} with the graph \ac{tv} filter from \eqref{eq; GHPF TV},  both  with $\Lmat=\Bmat$.
 The detector in \eqref{eq; detection GFT2} can be  extended and used  for the \ac{ac} model, as explained in Subsection \ref{sec; sim gfdi analysis}. 

\section{Graph False Data Injection (\ac{gfdi}) Attacks} \label{sec; gfdi}
  
  The graph-based detection methodology presented in Section \ref{sec; background}, 
  which provides the \ac{psse} approach with an additional defense layer against \ac{fdi} attacks,
  is illustrated in  Fig. \ref{fig; FDI schematic}.
In this section, we demonstrate how an adversary could use graph-based information 
to design an attack that is concentrated in the spectral region of the small eigenvalues
and, thus, is more likely to bypass the GSP detection methods (an illustration is provided in Fig. 
\ref{fig;  gfdi example}).   
In particular, in Subsection \ref{sec; attack design}, we formulate the \ac{gfdi} attack as a constrained optimization problem. 
Then, in Subsection \ref{sec; implementation}, we derive the solution for the \ac{gfdi} attack optimization problem.
Finally, some remarks are given in Subsection \ref{sec; remarks}. 
    
     \begin{figure}[tbt]
     	\centering
     	\includegraphics[trim={12.5cm 3.5cm 14.5cm 3.5cm},clip,width=3.3cm]{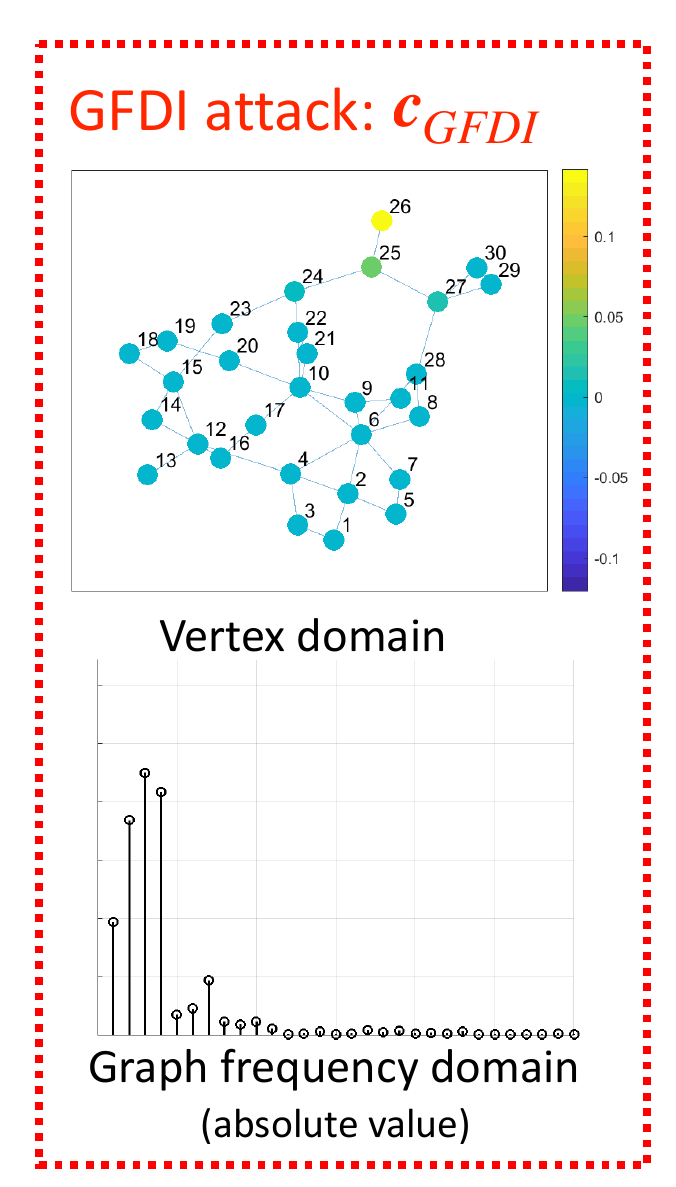}
     	\caption{Example of a \ac{gfdi} attack on the IEEE $30$-bus test case from Fig. \ref{fig; FDI schematic}.
              It can be seen that the attack energy is located at the lower graph frequencies. 
             Therefore, in contrast to the unobservable \ac{fdi} attack in Fig. \ref{fig; FDI schematic}[b], when added to the system states, the output will not obtain abnormal energy in the high graph frequencies and, thus, may bypass the \ac{gsp}-based detectors.}
     	\label{fig;  gfdi example}
     \end{figure}
     \setlength{\textfloatsep}{5pt}
     \vspace{-0.8cm}
\subsection{Attack design} \label{sec; attack design}
The proposed \ac{gfdi} attack is a special case of unobservable \ac{fdi} attacks specifically designed to bypass the graph-based detector in \eqref{eq; detection GFT2}.  
The main idea is to design the attack as the output of a smooth fake state signal.
Mathematically, let $\avec=\Hmat\cvec$ from \eqref{eq; FDI attack} be the attack vector, and $\Lmat=\Bmat$ be the graph Laplacian matrix. 
Note that $\Bmat$ is the \ac{dc} nodal admittance matrix, which is a submatrix of $\Hmat$. 
The state \ac{fdi} attack, $\cvec$, is encouraged to be a smooth signal w.r.t. the graph $\mathcal{G}$, i.e., to minimize the graph \ac{tv}, which according to \eqref{eq; TV} satisfies
\be \label{eq; TV c}
\begin{aligned}
TV^{\Gcal}(\cvec)= \cvec^T \Lmat  \cvec. 
\end{aligned}
\ee
Alternatively, as explained in Remark \ref{rem; graph filter}, the graph \ac{tv} measure in \eqref{eq; TV c}
can be replaced by its generalization in \eqref{eq; detection GFT}, 
which is used for the \ac{gsp}-based detection in \eqref{eq; detection GFT2}. 
 Simultaneously, the attack should have a substantial impact on the state estimation in order to cause damage.
	  Therefore, in order to make the attack ``meaningful'', the adversary further constrains that \cite{kim2011strategic}
	   \be \label{eq; c impact}
	  \tlvert \cvec \trvert_{\infty} \ge \tau,
	  \ee
 with some positive threshold, $\tau$.
 In other words, the deviation in the \ac{psse} output in \eqref{eq; hypothesis power} caused by 
	  at least one element in $\cvec$ should be larger than the threshold $\tau$. 
	  
	  In addition, the attack design considers practical limitations. 
	  First, as in other \ac{fdi} attacks (see, e.g., \cite{liu2011false, kosut2011malicious}), the attack should be sparse.
      Specifically, similar to previous works (see, e.g.,  \cite{gao2016identification,ozay2012distributed,morgenstern2022structural,tian2022exploring,kim2011strategic,liang2017false}), the sparsity constraint is applied directly on the state attack vector $\cvec$,  i.e., the number of the manipulated state variables is considered to be small. 
	  For a sparsity parameter $k\in \mathbb{R}_+$, the attack is constrained by 
	  \be \label{eq; c sparsity} 
	  \tlvert \cvec \trvert_{0} \le k,
	  \ee
	  where $\tlvert \cdot \trvert_0$ is the $\ell_0$ semi-norm defined as the number of nonzero elements of its argument. 
	
Simulations performed in \cite{liu2011false}   show that the assumption in \eqref{eq; c sparsity} stems directly from the commonly-used sparsity restriction on the number of manipulated meters, which states that the attack vector, $\avec$, is sparse.

      Finally, it is assumed that certain security constraints are imposed by the system designer. 
	  Specifically, we assume that a set of indices corresponding to the measurement set is protected. 
	  The measurement constraint for the attacker can be expressed as 
	  \be \label{eq; c restricted}
      \Hmat^{\Scal}\cvec=\zerovec. 
      \ee  
	 In the above, we denote by $\Hmat^{\Scal}$ the matrix formed by the $\lvert \Scal \rvert$ rows of $\Hmat$ indicated by the indices in $\Scal$.  The constraint in \eqref{eq; c restricted} implies that the measurements in the set $\Scal$ do not contribute to the attack vector $\avec$ in
  \eqref{eq; FDI attack}.

	 To conclude, the attack has two main conflicting goals: to be less detectable by GSP tools (as described in \eqref{eq; TV c}), 
	 while causing a significant impact on the power system (as described in \eqref{eq; c impact}). These two goals should be achieved while adhering to 
	 the physical constraints on the possible attacked locations: 
	 a quantitative constraint (described  in \eqref{eq; c sparsity}) 
	 and a qualitative constraint on the specific sensors (described in \eqref{eq; c restricted}).  
	 This can be formalized by the following optimization problem:
	     \be  \label{eq; gfdi start}
	  \begin{aligned}
	   \underset{\cvec\in\mathbb{R}^{N}}{\min} ~ \cvec^T \Lmat \cvec \\
	  	&\hspace{-0.5cm}\text{s.t.} \begin{cases}
	  		\tlvert \cvec \trvert_\infty \ge \tau \\
	  		\tlvert \cvec \trvert_0 \le  k \\
	  		\Hmat^{\Scal}\cvec=\zerovec.
	  	\end{cases}
	  \end{aligned}
	  \ee
	  Roughly speaking, the proposed \ac{gfdi} attack in \eqref{eq; gfdi start} is the smoothest attack possible
	  that ensures sufficient damage and considers 
      practical limitations of sparsity and restricted measurements. It can be seen that without the extra constraint in \eqref{eq; c impact}, a trivial optimal solution to the optimization in \eqref{eq; gfdi start} is $\cvec=\zerovec$, which means that the
attacker does not attack the system.

In the following, we show the necessity of the constraints from the \ac{gfdi} optimization problem in \eqref{eq; gfdi start}.
\begin{enumerate}

    \item \textbf{No impact:}
The impact parameter $\tau$ must satisfy $\tau>0$.
This is since otherwise, i.e., if $\tau=0$, then the solution of \eqref{eq; gfdi start} 
is $\cvec=\zerovec$, which implies zero-attack ($\avec=\zerovec$). 
    
    \item \textbf{Sparsity:} The sparsity parameter $k$ must satisfy $k< N$. Otherwise,
for $k=N$ the solution of \eqref{eq; gfdi start}
is in the linear space spanned by the first eigenvector of the Laplacian matrix, 
i.e., $\cvec\in \mathrm{span}(\mathbf 1)$. 
 This can be seen from the fact that the smallest eigenvalue of the Laplacian is $0$, which guarantees that $\cvec^T\Lmat\cvec=0$. This 
 is the lowest value of \ac{tv} possible, since $\Lmat$ is a positive semidefinite matrix. 
Moreover, it can be verified from \eqref{eq; H flows} and \eqref{eq; H injections} that a solution  $\cvec\in \mathrm{span}(\mathbf 1)$ results in
an attack vector $\Hmat\cvec=\zerovec$. 
Thus, $\cvec\in \mathrm{span}(\mathbf 1)$ is a feasible solution since it satisfies 
$\Hmat^{\mathcal{S}}\cvec=\zerovec$. 
However, in this case, the system is not really attacked in the sense of $\avec$, and thus, this is a degenerative case. 
It should also be noted that in the case where the constraint $\tlvert \cvec \trvert_0 \le  k $ is replaced with $\tlvert \avec \trvert_0 \le  k$, the solution is bound to satisfy $\cvec\in \mathrm{span}(\mathbf 1)$ such that it has no impact on the actual states. Thus, in the considered setting, $\tlvert \cvec \trvert_0 \le  k $ is the appropriate sparsity constraint.

\item \textbf{No availability:} 
A feasible solution may not exist
when the secured set $\mathcal{S}$ includes a substantial number of sensors. 
In the extreme case where $\mathcal{S}$ contains all the system sensors, it is obvious 
that an \ac{fdi} attack is infeasible. 
However, if an unobservable \ac{fdi} attack 
that satisfies the sparsity and impact restriction in \eqref{eq; gfdi start} is available for a selected set $\mathcal{S}$, 
then a \ac{gfdi} attack is available as well, i.e., \eqref{eq; gfdi start} has a solution.

\end{enumerate}
\vspace{-0.5cm}
\subsection{GFDI attack implementation} \label{sec; implementation}
         
 The optimization problem in \eqref{eq; gfdi start} is composed of a quadratic objective function, $ \cvec^T \Lmat \cvec$,
  with: 1) a concave inequality constraint, $\|\cvec\|_{\infty}\ge\tau$, since the $\ell_{\infty}$-norm  is a convex function \cite{boyd2004convex}; 
  2) a non-convex sparsity constraint, $\tlvert \cvec \trvert_0$; and 3) a linear constraint, $\Hmat^{\Scal}\cvec=\zerovec$.
  Hence, \eqref{eq; gfdi start} is a non-convex optimization problem.
 In this subsection, we derive a solution for this problem. 
 
The following theorem suggests an equivalent optimization problem for the \ac{gfdi}  attack optimization problem in \eqref{eq; gfdi start}. 
\begin{thm} \label{thm; gfdi cases} 
	The solution of the \ac{gfdi} attack optimization problem in \eqref{eq; gfdi start}, $\hat{\cvec}$, can be obtained by solving the following series of optimization
problems:
	\be  \label{eq; gfdi cases}
	\begin{aligned}
		\underset{i=\{1,\ldots,N\}}{\min}  	\underset{\cvec\in\mathbb{R}^{N}}{\min} ~ \cvec^T \Lmat \cvec \\
		&\hspace{-0.5cm}\text{s.t.} \begin{cases}
			 c_{i}=\tau \\
            \tlvert \cvec \trvert_0 \le  k\\
             \Hmat^{\Scal} \cvec=\zerovec.			
		\end{cases}
	\end{aligned}
	\ee
\end{thm}
\begin{proof} 
	The proof is given in Appendix \ref{App; cases}. 
\end{proof}
It can be seen that in the inner minimization in \eqref{eq; gfdi cases}, the $\ell_{\infty}$ inequality constraint from \eqref{eq; gfdi start} is replaced with the $i$th linear constraint, $c_i=\tau$.
In addition, it can be seen that the inner optimization problem in  \eqref{eq; gfdi cases} is composed of a quadratic objective function with sparsity and linear constraints. 
Hence, the major issue for the attacker is that solving optimization problems with $\ell_0$ constraints is, in general,  NP-hard. 
Following standard sparse recovery techniques \cite{elad2010sparse}, the $\ell_0$-norm can be replaced by its $\ell_1$-norm relaxation version 
that promotes sparsity. Thus, the attacker can solve the following 
convex relaxation:
\be  \label{eq; gfdi cases inner l1}
\begin{aligned}
	\underset{\cvec\in\mathbb{R}^{N}}{\min} ~ \cvec^T \Lmat \cvec \\
	&\hspace{-0.5cm}\text{s.t.} \begin{cases}
		    c_{i}=\tau \\
                \tlvert \cvec \trvert_1 \le  k\\
			 \Hmat^{\Scal} \cvec=\zerovec.	
	\end{cases}
\end{aligned}
\ee
That is, instead of solving \eqref{eq; gfdi start}, the adversary should solve a series of $N$ convex optimization problems 
obtained by setting $i=\{1,\ldots,N\}$ in \eqref{eq; gfdi cases inner l1}. 
The final step comprises selecting the minimum from the $N$ solutions.

By introducing the nonnegative vector variables $\uvec\ge \zerovec$ and $\vvec\ge \zerovec$, where the vector inequality indicates elementwise inequalities, 
 the problem in \eqref{eq; gfdi cases inner l1} can be formulated as a quadratic programming problem, as follows:
\be  \label{eq; gfdi QP}
	\begin{aligned}
		 	\underset{\cvec\in\mathbb{R}^{N}}{\min} ~ \cvec^T \Lmat \cvec \\
		&\hspace{-0.5cm}\text{s.t.} \begin{cases}
			 c_{i}=\tau \\
             \Hmat^{\Scal} \cvec=\zerovec \\
             \cvec=\uvec-\vvec \\
             \onevec^T(\uvec+\vvec)\le k \\
             \uvec\ge \zerovec, \vvec\ge \zerovec,
		\end{cases}
	\end{aligned}
 \ee
 where $\onevec$ is the all-one vector.
The quadratic programming problem in \eqref{eq; gfdi QP} can be efficiently solved using interior point methods\cite{boyd2004convex}, e.g., using the Matlab function \textit{quadprog}. 

After solving \eqref{eq; gfdi QP}, if a feasible solution is found, then case $i$ is added to the set of possible solutions,  $\mathcal{I}$,
 and the solution is denoted by $\cvec^{*,i}$. 
In order for the solution  $\cvec^{*,i}$ to be at most $k$-sparse, as required by the original constraints of the problem in \eqref{eq; gfdi cases}, we apply a hard-thresholding step and only keep the $k$  largest  (in the sense of the absolute value of the magnitudes) components of  $\cvec^{*,i}$, while zeroing out
the remaining entries. 
On the other hand, if a feasible solution is not available, then case $i$ is not included in $\mathcal{I}$.

After solving \eqref{eq; gfdi QP} for $i=1,\ldots,N$, we determine the optimal position to attack the system by choosing the optimal solution among the candidates in $\mathcal{I}$:
\be
\label{hat_i}
\hat{i}=\arg\underset{i \in \mathcal{I}}{\min}~(\cvec^{*,i})^T \Lmat \cvec^{*,i}.
\ee

Next, we
calculate the optimal attack by setting $\hat{\avec}=\Hmat\hat{\cvec}$, where $\hat{\cvec}=\cvec^{*,\hat{i}}$ and $\hat{i}$ is given in \eqref{hat_i}.
However, due to the thresholding step applied after solving \eqref{eq; gfdi QP}, there is no guarantee that the resulting $\hat{\cvec}$
maintains the constraint on the secured sensors from \eqref{eq; gfdi QP}, i.e., there is no guarantee that $\hat{\avec}_{\mathcal{S}}=\Hmat^{\Scal}\hat{\cvec}=\zerovec$.
Therefore, in order for the final solution $\hat{\cvec}$ to be a feasible solution, we set the elements in $\hat{\avec}$ that correspond to the set $\mathcal{S}$ to zero, i.e., we set
\be \label{eq; a S}
\hat{\avec}_{\mathcal{S}}=\zerovec.
\ee

\begin{algorithm} [ht]
\caption{\ac{gfdi} attack creation}
\label{alg; gfdi}
\begin{algorithmic}[1]
\renewcommand{\algorithmicrequire}{\textbf{Input:}}
\renewcommand{\algorithmicensure}{\textbf{Output:}}
\REQUIRE ~ $\Lmat$, $k$, $\tau$, $\Scal$,  $\Hmat$  \vspace{0.1cm}
\ENSURE  $\hat{i}$, $\hat{\cvec}$,  and $\hat{\avec}$ 
\STATE  Initialize $\mathcal{I}=\emptyset$ 
\FOR {$i\in\{1,\ldots,N\}$}

\STATE solve \eqref{eq; gfdi QP} (e.g., by Matlab  \textit{quadprog})
\IF{a feasible solution, $\cvec^{*,i}$, is found }
 \STATE update $\mathcal{I}=\mathcal{I} \cup i$
\STATE set $c^{*,i}_m=0$ for $m>k$ and $\lvert c^{*,i}_{j_1}\rvert \ge \ldots \ge \lvert c^{*,i}_{j_N}\rvert$
 \ELSE
 \STATE go to line 2
 \ENDIF
\ENDFOR	
\IF{ $\mathcal{I}=\emptyset$}
\RETURN ``no feasible solution''
\ELSE
\STATE compute: $\hat{i}=\arg\underset{i \in \mathcal{I}}{\min}~ (\cvec^{*,i})^T \Lmat \cvec^{*,i}$ 

\STATE  set $\hat{\cvec}=\cvec^{*,{\hat{i}}}$ and $\hat{\avec}=\Hmat\hat{\cvec}$
\STATE  update $\hat{\avec}_{\mathcal{S}}=\zerovec$

\RETURN  $\hat{i}$, $\hat{\cvec}$,  and $\hat{\avec}$ 
\ENDIF
\end{algorithmic} 
\end{algorithm}

\vspace{-1cm}
\subsection{Remarks} \label{sec; remarks} 
\subsubsection{Attack Unobservability}

Due to Step $16$ in Algorithm \ref{alg; gfdi}, the attack $\hat{\avec}$ can be viewed as the following superposition:
\be
\hat{\avec}=\Hmat\hat{\cvec}+\tilde{\avec},
\ee
where the vector $\tilde{\avec}$ is such that $\tilde{\avec}_{S}=\hat{\avec}^{S}-\Hmat^{\mathcal{S}}\hat{\cvec}$ and $\tilde{\avec}_{\{\mathcal{M}\setminus \mathcal{S}\}}=\zerovec$.
Consequently, the proposed attack cannot be considered as a pure unobservable attack in the sense of \eqref{eq; FDI attack}. 
Nonetheless, we expect that $\|\tilde{\avec}\|\le \epsilon$, where $\epsilon$ is a significantly low value.  
This implies that 
$$
\|\hat{\avec}-\Hmat\hat{\cvec}\|_2^2\le \epsilon.
$$ 
Hence, this attack can be referred to as a generalized unobservable \ac{fdi} attack (see Section 4.1. in  \cite{liu2011false}) and is still expected to be undetected by classical residual-based \ac{bdd} methods if a certain degree of noise is present.

\subsubsection{Known topology}
An assumption behind 
the \ac{gfdi} attack design in \eqref{eq; gfdi start} is that $\Lmat$ is known.
In other words, it is required that the adversary knows the power system network configuration. 
The same assumption is required for generating the unobservable FDI attack in \eqref{eq; FDI attack} \cite{liu2011false,kosut2011malicious}. 
The assumption that $\Hmat$ (and consequently its submatrix $\Bmat$) is known or can be estimated from historical data \cite{kekatos2015online,kim2014subspace,grotas2019power,Halihal_Routtenberg_2022},
gives the adversary more power than usually is possible in reality. 
We note that this is a well-adopted practice in the cyber-security community, which increases the system's resilience.

\subsubsection{Generalization to graph filters} \label{rem; graph filter}
The objective function of the \ac{gfdi} optimization problem in \eqref{eq; gfdi start}
can be generalized by replacing the graph \ac{tv} measure in \eqref{eq; TV c} 
with a general \ac{gsp}-based smoothness measure in \eqref{eq; detection GFT2}, 
which measures the energy of the estimated \ac{psse} after it has been filtered by a selected \ac{ghpf}. 
In this case,  the quadratic programming optimization in \eqref{eq; gfdi QP} 
is modified by replacing the cost function 
with $f(\cvec)=\|f(\Lmat)\cvec\|^2$, 
where $f(\Lmat)$ is a general graph filter as defined in \eqref{eq; graph filter decomposition}.
In particular, we can obtain the following two special cases:  
1) if the \ac{ghpf} is selected as \eqref{eq; TV},
then we obtain the same result provided by the graph \ac{tv} measure, described above and summarized in Algorithm \ref{alg; gfdi}; 
and 2) if the \ac{ghpf} is selected as \eqref{eq; GHPF ideal}, then the attack obtained will be a graph low pass signal with energy located only in the $k$ smallest graph frequencies, as conceptualized in \cite{shereen2022detection}.
Thus, the proposed approach also suggests an implementation for the theoretical idea in \cite{shereen2022detection}.

\subsubsection{Applicability of the \ac{gfdi} attack for the \ac{ac} model}  \label{rem; AC}
The nonlinear \ac{ac}
provides a more accurate representation of the power flow equations 
than the \ac{dc} model presented in Section \ref{sec; background} \cite{kosut2011malicious}.  
As described in Section \ref{sec; background},  constructing an unobservable \ac{fdi} attack for the \ac{dc} model 
 required knowledge of the system configuration represented by $\Hmat$.
In contrast, constructing an unobservable \ac{fdi} attack for the \ac{ac} model requires knowledge of both $\Hmat$ 
and the current state of the system \cite{kosut2011malicious}, which is often considered unrealistic. 
However, since the \ac{dc} model is a linearization of the
\ac{ac} model, a \ac{dc}-based attack could work approximately on the \ac{ac} model \cite{kosut2011malicious}. 
We show the applicability of our attack design on both the \ac{dc} and the \ac{ac} models in the simulation study in Section \ref{sec; simulations}.

\section{Strategic protection of the power network}  \label{sec; protection} 

In addition to \ac{gsp}-based detection, the operator 
may install additional security hardware to disable access to selected measurements for an adversary. 
In this section our goal is to demonstrate how to harvest the graph-based knowledge,
 used for the design of the \ac{gfdi} attack in Section \ref{sec; gfdi}, in order to strategically select the protected measurements.
 Specifically, we aim to identify the minimal set of measurements required to prevent the possibility 
 of a \ac{gfdi} attack. 
This section is organized as follows. In Subsection \ref{sec; protection design} 
we discuss the protection scheme design. Then, in Subsection \ref{sec; protection implemenation} 
we provide a practical implementation. We conclude with remarks in Subsection \ref{sec; protection remarks}. 

\vspace{-0.4cm}
\subsection{Protection scheme design} \label{sec; protection design}
In practice, our protection scheme identifies a minimal set of {\em{state variables}} (denoted by $\Dcal$), such that if secured (not manipulated), it would disable the possibility of generating 
a \ac{gfdi} attack. 
After identifying $\Dcal$, we can then deduce the set of secured {\em{measurements}} (denoted by $\mathcal{S}$) in one of the following two ways:
1) the set $\mathcal{S}$ is composed from the power measurements positioned in  
the buses in set $\Dcal$ (power injections) and in the lines entering these buses 
(power flows); and 2) the set $\mathcal{S}$ is composed from \ac{pmu} measurements installed at the same locations as in set $\Dcal$.
In our implementation, we adopt the first option. 

The proposed protection scheme is formulated by the following  heuristic  optimization problem:
\be \label{eq; protection}
\begin{aligned}
\hat{\Dcal}= \arg	\underset{\Dcal\subseteq\Vcal}{\min}~\lvert\Dcal\rvert~
s.t.~(\hat{\cvec})^T\Lmat\hat{\cvec}>\delta, 
\end{aligned}
\ee
where $\hat{\cvec}$ is the output of the \ac{gfdi} attack optimization problem in \eqref{eq; gfdi start}, which is implemented by Algorithm \ref{alg; gfdi}.   
Thus, the problem in \eqref{eq; protection} searches for the set $\Dcal$ with the minimal cardinality, 
for which the graph \ac{tv} of the optimal \ac{gfdi} solution exceeds a user-defined threshold $\delta>0$, and therefore, can be detected by GSP methods. 
Selecting a minimal set of secured state variables enables the operator to reduce the installation cost of security hardware.

The underlying assumption behind the strategic protection design in \eqref{eq; protection} is its effectiveness against the proposed \ac{gfdi} attack, along with its capability to provide a defense against various potential smooth attack vectors that may not be the optimal \ac{gfdi} attack. Furthermore, this framework can be extended to other GSP-based attacks described in Remark \ref{rem; graph filter} by using the generalization to other GHPFs. Future research is needed to identify an optimal protection strategy that can defend the system against a broad spectrum of smooth attacks.

\vspace{-0.4cm}
\subsection{Implementation} \label{sec; protection implemenation}
The problem in \eqref{eq; protection} is a combinatorial optimization problem
with a non-submodular objective function. 
Thus, the number of possible instances for $\Dcal$ grows exponentially with the system size,  $N$. 
Moreover, it can be seen that, for each instance, the objective function requires
implementing Algorithm \ref{alg; gfdi}.
Thus, \eqref{eq; protection} will suffer from high computational complexity
due to the exhaustive search and is not practical for large networks. 
Therefore,  we propose a low-complexity greedy algorithm for selecting the state variables to be protected, as described in Algorithm \ref{alg; protection}.

We start with an empty set, $\mathcal{D}=\emptyset$, and iteratively compose $\Dcal$ as follows.
In each iterative step, the \ac{gfdi} problem in \eqref{eq; gfdi cases} is first solved 
by Algorithm \ref{alg; gfdi}, taking into account the set $\Scal$ 
composed from the power measurements positioned at 
the buses in the current set $\Dcal$ (power injections) and in the lines entering these buses 
(power flows). 
Then, the state variable at position $i$, where the attack obtains the 
maximal value $\tau$, is added to the secured state variable set $\mathcal{D}$. 

\begin{algorithm} [h!]
\caption{Protection scheme}
\label{alg; protection}
\begin{algorithmic}[1]
    \renewcommand{\algorithmicrequire}{\textbf{Input:}}
    \renewcommand{\algorithmicensure}{\textbf{Output:}} 
    \REQUIRE  $\Lmat$, $k$, $\tau$, $\delta$
    \ENSURE  $\Dcal$
    \STATE Initialize: $\Dcal=\emptyset$ 
    \REPEAT
      \STATE derive $\Scal$ from $\Dcal$ by including power injections in the buses in $\Dcal$ and power flows entering the same buses
    \STATE\label{stage3} get $\hat{i}$ and $\hat{\cvec}$ from \ac{gfdi} Algorithm  
    \STATE add $\hat{i}$ to $\Dcal$
    \UNTIL{$(\hat{\cvec})^T\Lmat \hat{\cvec}> \delta$}
    \RETURN  $\Dcal$
\end{algorithmic} 
\end{algorithm}	

   \vspace{-0.6cm}
\subsection{Remarks} \label{sec; protection remarks}
\subsubsection{Generalization to graph filters} \label{rem; protection; graph filter}
If the \ac{gfdi} optimization problem is realized 
with a general graph filter as explained in Remark $2)$ in Subsection \ref{sec; remarks}, 
then the constraint in \eqref{eq; protection} should be replaced accordingly. 
The following is conducted by replacing the graph \ac{tv} measure in \eqref{eq; protection} 
with $\|f(\Lmat)\cvec\|^2$, which is obtained by substituting $\svec=\hat\cvec$ in \eqref{eq; detection GFT}. 
In this case, row $4$ of Algorithm \ref{alg; protection} calls Algorithm \ref{alg; gfdi} with the changes discussed in Remark $2)$ in Subsection \ref{sec; remarks}.

\subsubsection{Comparison to previous protection schemes} \label{rem; protection; previous}
The proposed protection scheme differs from previous designs in that it does not aim to prevent the adversary from launching an \ac{fdi} attack. Rather, it seeks to eliminate the possibility of an 
 attack with low graph \ac{tv}. 
 Consequently, the proposed designs forces the \ac{gfdi} output to be a non-smooth attack, which, with high probability, would be detected by a \ac{gsp}-based detector. As a result, there are two minor drawbacks and one significant advantage.
The first drawback is that a \ac{gsp}-based detector must be installed in the system control center. However, since these detectors are software-based, they can be implemented easily.
The second drawback is that if an attack is detected, the system still operates with unreliable measurements, which is currently a common limitation of detection policies.
On the other hand, the relaxation provided by enabling a \ac{gfdi} attack instead of no \ac{fdi} attack at all significantly reduces the number of secured measurements required, as demonstrated in the simulation study (see Fig. \ref{fig; protection}).

\vspace{-0.2cm}
\section{Simulation study} \label{sec; simulations}

This section demonstrates the performance of the proposed \ac{gfdi} attack and GSP-based protection policy by numerical simulations.
The simulations are conducted on the IEEE-$57$ and IEEE-$118$ bus test cases,
where the topology matrix and measurement data are extracted using the Matpower toolbox for Matlab \cite{matpower}.
In Subsection \ref{sec; sim methods}, we describe the \ac{fdi} attack construction methods, 
and protection policies used as reference.  
The simulation setup is then provided in Subsection \ref{sec; set up}. 
Next, we present the analysis conducted for the 
 \ac{gfdi} attack in Subsection \ref{sec; sim gfdi analysis}
 and for the \ac{gsp}-based protection policy in Subsection \ref{sec; sim protection analysis}. 

\subsection{Methods} \label{sec; sim methods}

\subsubsection{\ac{fdi} attack constructions}
\label{sec; sim; FDI constructions}
The \ac{gfdi} attack (denoted as \texttt{GFDI}) implemented by Algorithm \ref{alg; gfdi} is compared to previous unobservable \ac{fdi} attacks. 
Considering that an unobservable \ac{fdi} attacks satisfies 
$\avec=\Hmat\cvec$, where $\Hmat$ is a known parameter, 
the attack can be equivalently defined using the state attack $\cvec$. 
The previous designs include the following:
\begin{enumerate}[label=\textbf{A.\arabic*}, leftmargin=1cm]
\item\label{Attack; random} 
    Random unobservable \ac{fdi} attack \cite{liu2011false} (denoted as \texttt{rand}): 
    This attack is constructed by: 1) selecting $k$ elements randomly for the support of the state attack, $\cvec$; and 2) assigning random values at these index locations
    according to the Gaussian distribution, $ \mathcal{N}(0,1)$. 
\item 
    Random unobservable \ac{fdi} attack with \ac{gfdi} support  (denoted as \texttt{rand+\ac{gfdi}}): 
        The state attack support is the one obtained from the \ac{gfdi} solution in Algorithm \ref{alg; gfdi}, where
     random values are assigned to the selected indices following the Gaussian distribution, $\mathcal{N}(0,1)$. 
     \label{Attack; random gfdi}
    \item 
    The sparsest unobservable \ac{fdi} attack with the lowest graph \ac{tv} (denoted as \texttt{sparse-low}): 
    The attack defined by the linear programming problem in Equation (14) in \cite{kim2011strategic} solved using the Matlab function \textit{linprog}. 
    If  there is more than one solution, 
    the one with the lowest graph \ac{tv} is chosen as the attack.
\label{Attack; sparsest lowest}
    \item 
    The sparsest unobservable \ac{fdi} attack with average graph \ac{tv} (denoted as \texttt{sparse-avg}):  The attack defined by the linear programming problem in Equation (14) in \cite{kim2011strategic} solved using the Matlab function \textit{linprog}. 
    In this case, if there is more than one solution,
    the one closest to the 
    average graph \ac{tv} is chosen.
   \label{Attack; sparsest average}
\end{enumerate}
The attacks are scaled to satisfy $\tlvert \cvec \trvert_{\infty}=\tau$.

\subsubsection{Protection policies}
The proposed \texttt{GSP-based protection} policy, suggested in Section \ref{sec; protection} 
in the presence of a \ac{gfdi} attack, is compared with the following protection policies:
\begin{enumerate}[label=\textbf{P.\arabic*}, leftmargin=1cm]
 \item 	  \texttt{Random-based protection} policy, where the elements in $\mathcal{D}$ are selected randomly. \label{P; random}
 \item    \texttt{Sparsest-based protection} policy, where the elements in $\mathcal{D}$ are selected to increase the number of nonzero elements in $\cvec$ required for a feasible unobservable \ac{fdi} attack described by Equation (14) in \cite{kim2011strategic}.
 The policy is defined in Algorithm $3$ in \cite{kim2011strategic}.\label{P; sparsest}
\end{enumerate}

\vspace{-0.3cm}
\subsection{Simulation setup} \label{sec; set up} 

The simulations were conducted on the IEEE-$57$ and IEEE-$118$ bus test cases. 
Detection thresholds were computed from simulated historical data obtained by $10,000$ off-line simulations of \eqref{eq; DC model} under the null hypothesis.
 All numerical results were obtained using $10,000$ Monte Carlo simulations.
 The \ac{fdi} attacks in both settings are modeled by the \ac{gfdi} attack 
 and Attacks \ref{Attack; random}-\ref{Attack; sparsest average}.

 \subsubsection{\ac{dc} model}
For the \ac{dc} model, the measurements are computed using \eqref{eq; DC model}. 
The state vector, $\xvec$, is assumed to be a smooth graph signal that has a low graph \ac{tv}, as defined in \eqref{eq; TV}, 
and is modeled as follows \cite{dieci1999smooth,dong2016learning}. 
In particular, 
we first set $\tilde{x}_1=0$ and generate
\be \label{eq; sim smooth phase}
\tilde{\xvec}_{2:end}\sim \mathcal{N}(\zerovec,\beta\boldsymbol{\Lambda}^{-1}_{2:end,2:end}),
\ee
where $\beta$ is a smoothness level selected as $\beta=0.05$. 
Then, we use the \ac{igft} defined below \eqref{eq; GFT} to obtain $\xvec=\Umat\tilde{\xvec}$. 

The measurement noise is modeled as a zero-mean Gaussian noise with variance $0.001\Imat$. 
In addition, one of the additive attacks presented in Subsection \ref{sec; sim methods} is added to the measurements. 

\subsubsection{\ac{ac} model}
For the \ac{ac} model, the measurements are modeled by
\be \label{eq; AC power flow model}
\zvec_{AC}= \vvec(\Yvec \vvec)^*+\avec+\evec_c,
\ee
where $\vvec$ are the complex voltages located at the system buses and $\Ymat$ is the admittance matrix.
The voltage phases are generated according to \eqref{eq; sim smooth phase}.
The voltage magnitudes, which are assumed to be close to $1$, are generated according to 
\be \label{eq; sim abs}
\lvert\vvec\rvert \overset{i.i.d}{\sim} \mathcal{N}(1,0.01).
\ee
The matrix $\Ymat$ follows the structure in \eqref{eq; L}, but, in this case, the weights are complex \cite{drayer2019detection}. 
As discussed in Subsection \ref{sec; remarks}, 
the same attack used on the \ac{dc} model is also used on the \ac{ac} model. 
Hence, one of the attacks in Subsection \ref{sec; sim methods} is added to the measurements. 
The measurement noise, $\evec_c$, is modeled as a circularly-symmetric complex Gaussian vector with zero-mean and a variance of $0.001$, i.e., $\evec_{c}\sim \mathcal{CN}(\zerovec,0.001\Imat)$.

\begin{figure}[t!]
\centering
 \begin{subfigure}[b]{0.5\textwidth}
\centering
\includegraphics[width=7cm]{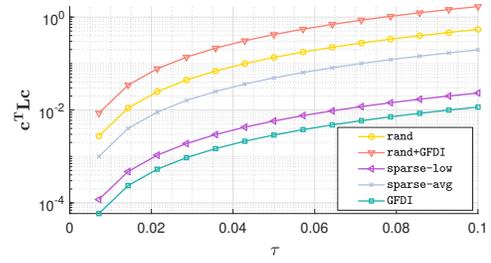}
\caption{Graph \ac{tv}}
\label{fig; tau GTV}
\end{subfigure}
\hfill 
 \begin{subfigure}[b]{0.5\textwidth}
\centering
\includegraphics[width=7cm]{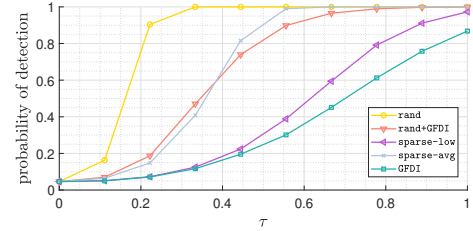}
\caption{Probability of detection}
\label{fig; tau pd}
\end{subfigure}
\caption{The \ac{gfdi} attack is compared to Attacks \ref{Attack; random}-\ref{Attack; sparsest average}. In (a)  and (b), the graph \ac{tv} of the attacks  and the probability of detection, respectively,  are presented versus $\tau$.  
The curves of the probability of detection are generated w.r.t. to the \texttt{GTV-GHPF} detector for the DC model. }
\label{fig; tau}
\end{figure}   

\vspace{-0.2cm}
 \subsection{\ac{gfdi} attack analysis} \label{sec; sim gfdi analysis}
 Figure \ref{fig; tau}.(a) presents the graph \ac{tv} of the
attacks as a function of the attack impact $\tlvert \cvec \trvert_{\infty}=\tau$ (see \eqref{eq; TV c}) for $k=5$  and $\Scal=\emptyset$
over the IEEE-$57$ bus test case. 
It can be seen that the \texttt{GFDI} attack  ($\hat{\cvec}$ from Algorithm \ref{alg; gfdi})  has the lowest graph \ac{tv}, 
while the other attacks generally demonstrate a rapid increase in the graph \ac{tv} as $\tau$ increases. 
These results can be interpreted as an indication of the vulnerability of the different attacks to the GSP-based detectors in the form of \eqref{eq; detection GFT2}. 
For this scenario, there is only one nonzero element in Attacks \ref{Attack; sparsest lowest}-\ref{Attack; sparsest average}. Thus, \ref{Attack; sparsest lowest} is the attack with the lowest \ac{tv} from all \ac{fdi} attacks limited to manipulate only one state variable. 
It can be seen that enabling the adversary to manipulate more than one element, as performed in the \texttt{GFDI} 
attack, results in better detection performance.

Figure \ref{fig; tau}.(b) presents the probability of detection as a function of the attack impact, $\tlvert \cvec \trvert_{\infty}=\tau$, for $k=5$  and $\Scal=\emptyset$ over the IEEE-$57$ bus test case. 
The detection is conducted by
the \texttt{GTV-GHPF} detector, which is obtained by applying \eqref{eq; detection GFT}
with the \ac{ghpf} in \eqref{eq; GHPF TV} and with a false alarm probability of $0.05\%$.
The results show that the probability of detection increases as $\tau$ increases, as expected.
Comparing the results in Fig. \ref{fig; tau}.(b) with those in  Fig. \ref{fig; tau}.(a) confirms the assumption that smooth attacks are harder to detect by \ac{gsp}-based detectors. 
As a result, the \ac{gfdi} ($\hat{\avec}$ from Algorithm \ref{alg; gfdi}) is the attack that is the hardest to detect.  
The \texttt{sparse-low} attack that is suggested in this paper takes the attack with the lowest graph TV from all the possible results provided by solving (14) in \cite{kim2011strategic}, and provides the closes probability of detection to the proposed \ac{gfdi} attack. 
However, it is important to note that the authors of \cite{kim2011strategic}, who introduced the sparsest attack, did not discuss the influence of graph \ac{tv} on detection. Thus, selecting the attack with the lowest graph \ac{tv},  i.e., the \texttt{sparse-low} attack,  can be seen as an additional contribution of this paper.  
Compared to the \texttt{GFDI} attack, the random \ac{fdi} attacks and the \texttt{sparse-avg} attack are easily detected.

  \begin{figure}[t!]
\centering
\begin{subfigure}[b]{0.5\textwidth}
\centering
 \includegraphics[width=5.5cm]{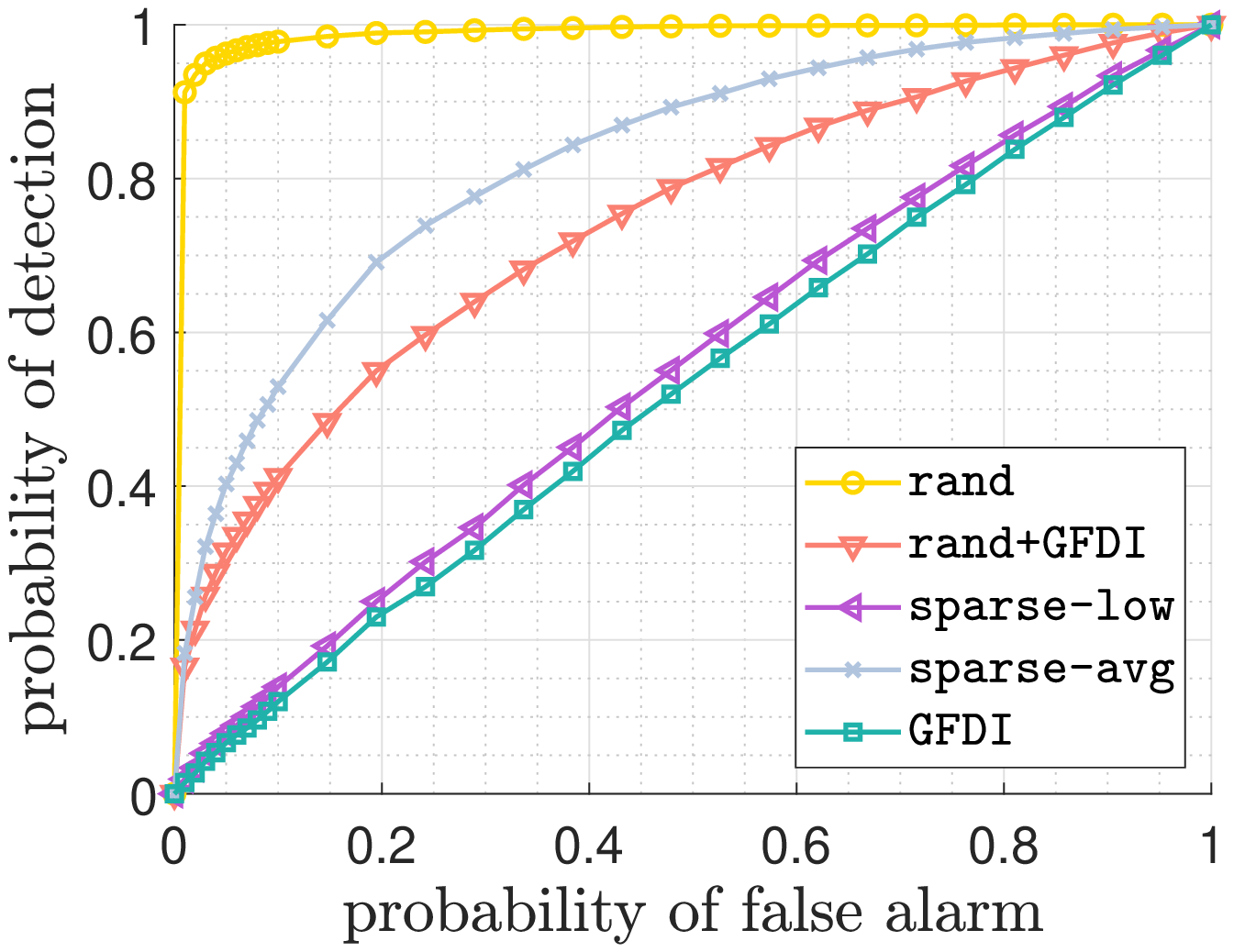}
 \caption{ \ac{dc} model}
\end{subfigure}
\hfill 
\begin{subfigure}[b]{0.5\textwidth}
\centering
  \includegraphics[width=5.5cm]{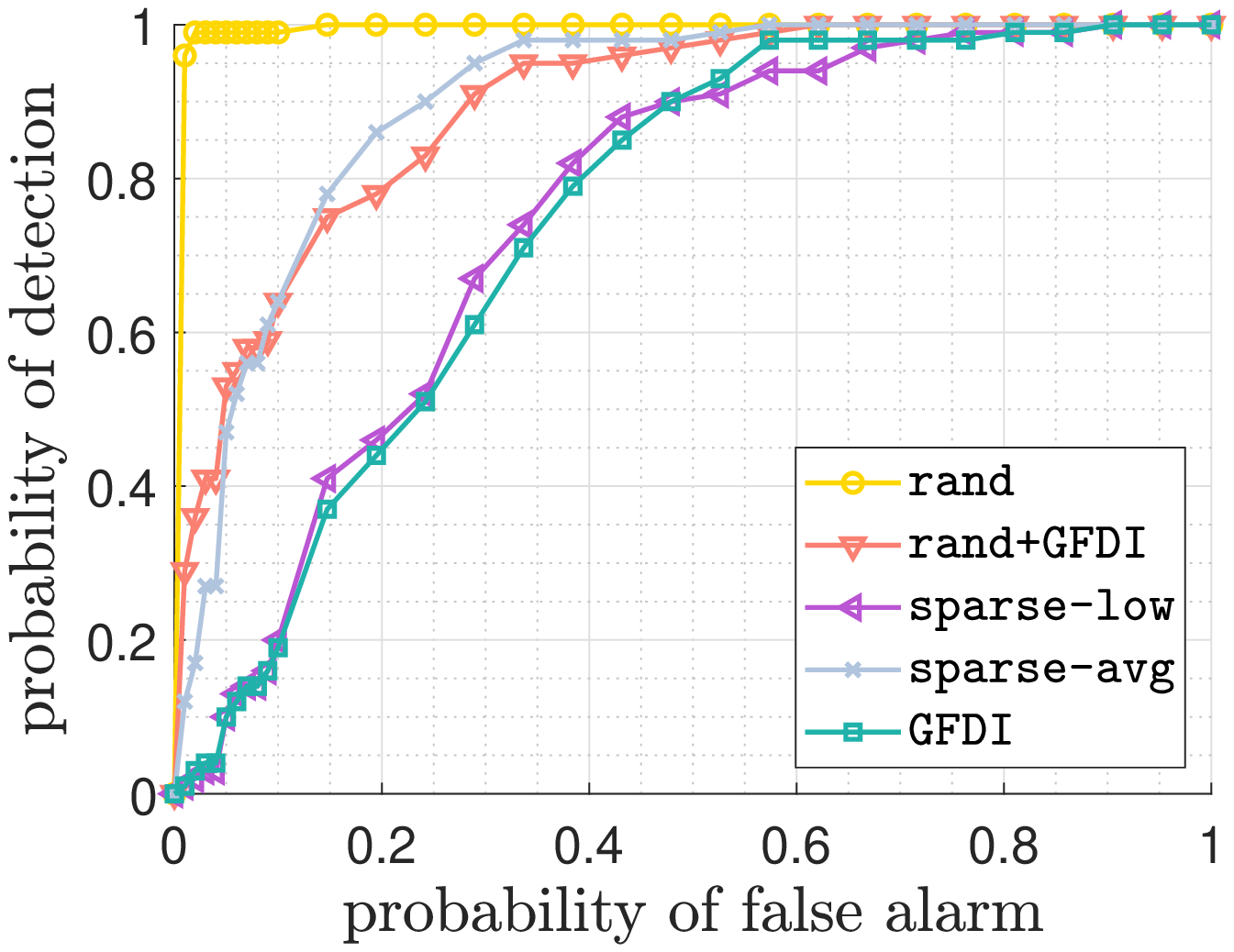}
 \caption{\ac{ac} model}
\end{subfigure}
\caption{The \ac{gfdi} attack is compared to Attacks \ref{Attack; random}-\ref{Attack; sparsest average}. In (a) and (b), the \ac{roc} curves of the attacks are generated w.r.t. to the \texttt{GTV-GHPF} detector for the DC model and the AC model, respectively.} 
\label{fig; roc}
\end{figure}

Figure \ref{fig; roc}.(a) presents the \ac{roc} curves (probability of detection versus probability of false alarm) of the \texttt{GTV-GHPF} detector
for the different unobservable \ac{fdi} attacks under the \ac{dc} simulation setup
with $\tau=0.2$, $k=5$, and $\mathcal{S}=\emptyset$ over the IEEE-$57$ bus test case.
It can be seen that the \texttt{GFDI} attack  ($\hat{\avec}$ from Algorithm \ref{alg; gfdi})  is the hardest to detect, 
while it performs slightly better than the \texttt{sparse-low} attack
Figure. \ref{fig; roc}.(b) presents the \ac{roc} curves of the \texttt{GTV-GHPF} detector
for the different unobservable \ac{fdi} attacks under the \ac{ac} simulation setup with $\tau=0.07$ and $k=5$ over the IEEE-$30$ bus test case.
It can be seen that the relationship between attack designs in 
Fig. \ref{fig; roc}.(b), where attacks with low graph \ac{tv} are less likely to be detected by the GSP-based detectors, is similar to those in Fig. \ref{fig; roc}.(a).
An extension of the \texttt{GTV-GFHP} to the \ac{ac} is provided as follows.
\begin{itemize}
\item \texttt{GTV-GHPF}: 
Based on our results in \cite{dabush2023state,Dabush_SAM_conf}, 
an alert to an attack is provided if the GSP-based detector in 
\eqref{eq; detection GFT2} when using the GTV-GHPF in \eqref{eq; GHPF TV} indicates an attack for at least one of the following cases:
1) $\Lmat=\Bmat$, $\yvec=\Re\{\hbt\}$, where $\hbt$ are the phases of $\hat{\vvec}$; and/or
2) $\Lmat=\Bmat$, $\yvec=\lvert\hat \vvec\rvert -\boldsymbol{1} \lvert \hat{v}_1\rvert$. 
\end{itemize}

Similar detection results to those presented in Figs. \ref{fig; tau}.(b),
\ref{fig; roc}.(a), and \ref{fig; roc}.(b)
were found when using the \texttt{Ideal-GHPF} 
detector, obtained by substituting  \eqref{eq; GHPF ideal} in \eqref{eq; detection GFT2}. 
In addition, we verified that the \texttt{BDD} detector cannot detect any of the attacks, since they are all unobservable.
Extensions of the \texttt{Ideal-GFHP} and \texttt{BDD} detectors for the \ac{ac} model can be obtained as follows.
\begin{itemize}
\item \texttt{Ideal-GHPF}: Attack alert is provided if the GSP-based detector in 
\eqref{eq; detection GFT2} when using the ideal GHPF in \eqref{eq; GHPF ideal} indicates an attack for at least one of the following cases: 
1) $\Lmat=\Re\{\Ymat\}$, $\yvec=\Re\{\hat{\vvec}\}$; 
2) $\Lmat=\Re\{\Ymat\}$, $\yvec=\Im\{\hat{\vvec}\}$;
3) $\Lmat=-\Im\{\Ymat\}$, $\yvec=\Re\{\hat{\vvec}\}$; and/or 
4) $\Lmat=-\Im\{\Ymat\}$, $\yvec=\Im\{\hat{\vvec}\}$.
The vector $\hat{\vvec}$ is the \ac{ac} \ac{psse} output. 
\item \texttt{BDD}:
The \texttt{BDD} detector for the \ac{ac} model ${\tlvert \zvec_{AC}-\hat{\vvec}(\Yvec \hat{\vvec})^*\trvert_2^2\gtrless_{\Hnull}^{\Ha}}$, 
where $\hat{\vvec}$ is the \ac{ac} \ac{psse} output. 
\end{itemize}

\begin{figure}[t!]
\centering
\begin{subfigure}[b]{0.5\textwidth}
 \centering
 \includegraphics[width=7cm]{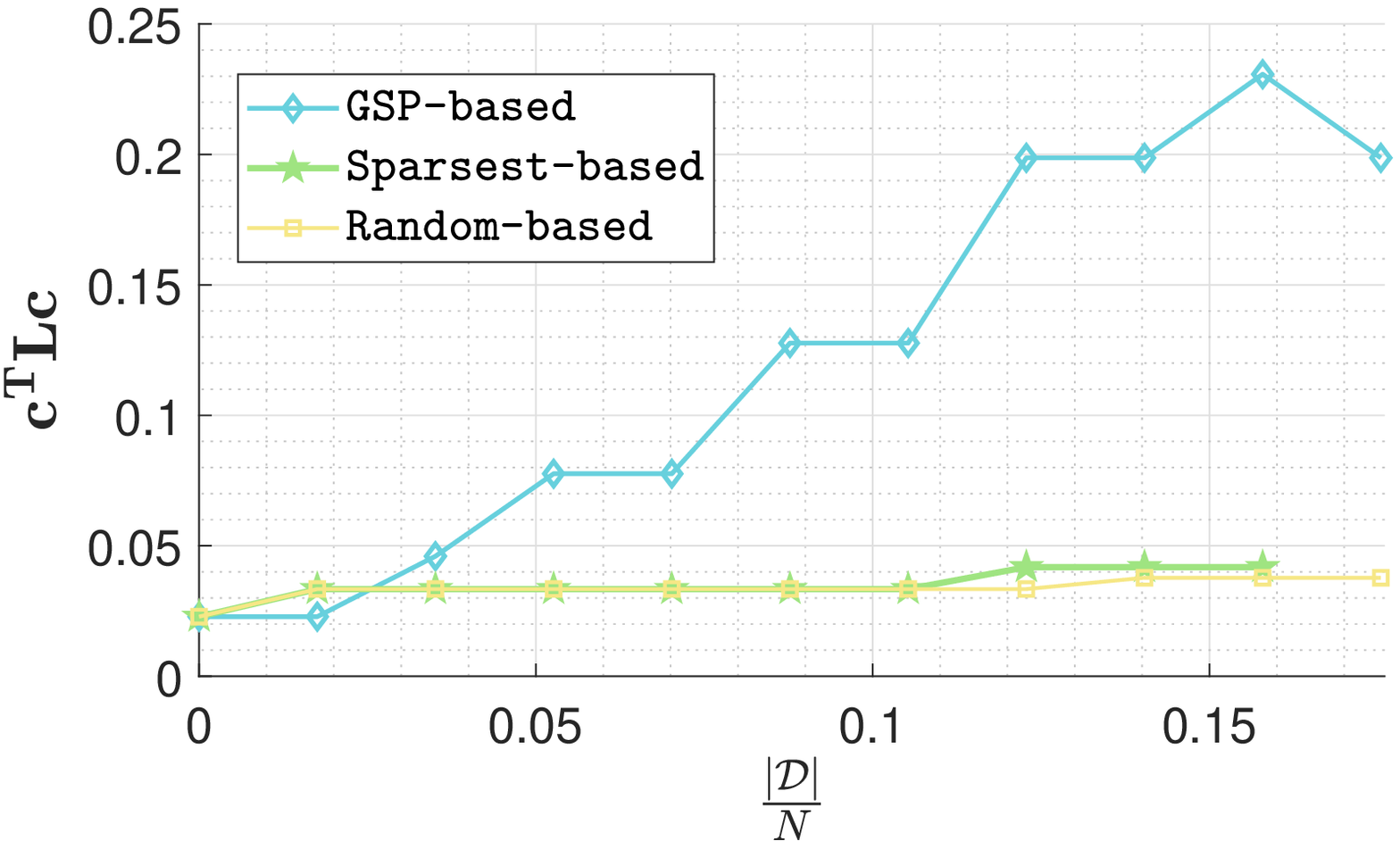}
 \caption{Protection policy analysis (IEEE-57)}
\end{subfigure}
\hfill 
\begin{subfigure}[b]{0.5\textwidth}
\centering
  \includegraphics[width=7.2cm]{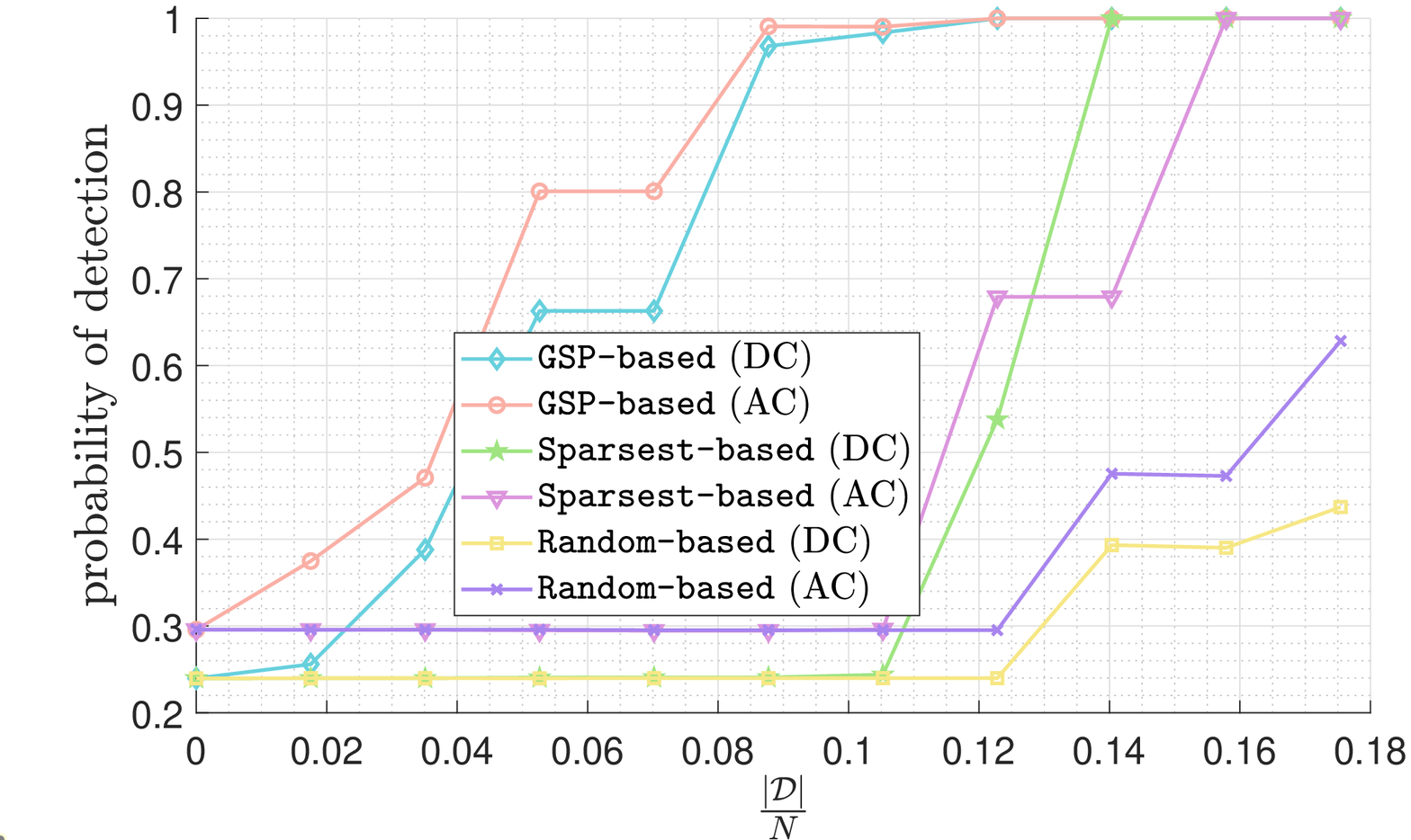}
 \caption{Protection impact on detection  (IEEE-57) }
\end{subfigure}
\begin{subfigure}[b]{0.5\textwidth}
\centering
  \includegraphics[width=7cm]{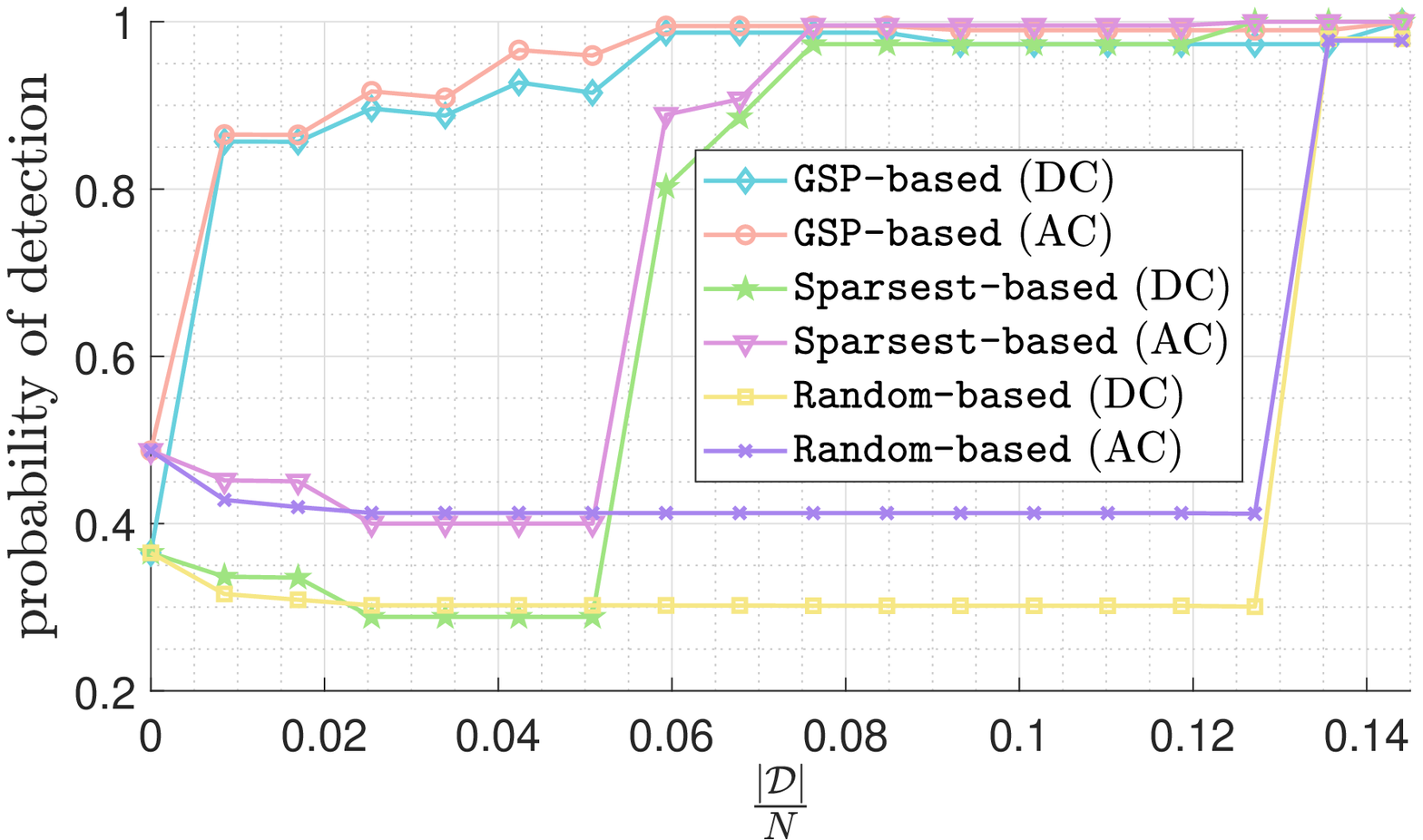}
 \caption{Protection impact on detection  (IEEE-118) }
\end{subfigure}
\caption{The \ac{gsp}-based protection policy 
is compared to Policies \ref{P; random}-\ref{P; sparsest}. 
 The IEEE-$57$ bus test case is observed in (a) and (b), where 
in (a), the \texttt{GFDI} attack graph \ac{tv} 
and in (b), the probability of detecting the \texttt{GFDI} attack by the \texttt{GTV-detector}, 
are presented versus the ratio of the secured state variables
out of the total number of states, $\frac{\lvert\mathcal{D}\rvert}{N}$.
In (c), the case in (b) is examined over the IEEE-$118$ bus test case. }
\label{fig; protection}
\end{figure}    

\vspace{-0.4cm}
\subsection{GSP-based protection policy analysis} \label{sec; sim protection analysis}
Figure \ref{fig; protection} illustrates the influence of 
securing system state variables, according to the different protection policies, on the 
\texttt{GFDI} attack.
 In Fig. \ref{fig; protection}.(a) the graph \ac{tv} of the \texttt{GFDI} attack  ($\hat{\cvec}$ from Algorithm \ref{alg; gfdi})  is presented versus the ratio between the secured state variables
 and the total number of states for the IEEE-$57$ bus test case
 with $\tau=0.6$ and $k=3$. 
The results indicate 
that using a \ac{gsp}-based policy leads to an immediate increase in graph \ac{tv} as the number of protected state variables increases.
In comparison, Policies \ref{P; random}-\ref{P; sparsest} have only a minor effect on the graph \ac{tv}.
In addition, it is important to mention that when at least $0.175\%$ of the state variables are secured, Policy \ref{P; sparsest} prevents the generation of the smooth attack.
Figure \ref{fig; protection}.(b) examines the probability of detecting the \texttt{GFDI}
attack by the \texttt{GTV-GHPF} detector,  with a false alarm probability of $0.05\%$,
for the IEEE-$57$ bus test case under both the \ac{dc} and \ac{ac} models. 
For the simulations conducted on the \ac{dc} and \ac{ac} models, we used $(\tau,k)=(0.6,3)$ and $(\tau,k)=(0.25,3)$, respectively.
It can be seen that protecting a low percentage of state variables according to the 
locations provided by a \ac{gsp}-based protection policy can significantly enhance the detection probability.
This is also true in Fig. \ref{fig; protection}.(c), which 
examines the probability of detecting the \texttt{GFDI}
attack by the \texttt{GTV-GHPF} detector,  with a false alarm probability of $0.05\%$,
for the IEEE-$118$ bus test case under both the \ac{dc} and \ac{ac} models. 
For the simulations conducted under the DC and AC models, we used $(\tau,k)=(0.55,5)$ and $(\tau,k)=(0.28,5)$, respectively.
It can be seen in Figs. \ref{fig; protection}.(b)  and \ref{fig; protection}.(c)  that protecting state variables according to the locations provided by a sparse-based protection policy also enhances the detection probability. 
However, this policy necessitates securing a significantly higher proportion of state variables. Moreover, it should be noted that the sparsest-based policy was created to prevent the possibility of an attack being generated, which requires, in general, securing at least $30\%$ of the state variables.
Finally, it can be seen that protecting state variables according to the locations provided by a random-based protection policy does not promote enhanced detection probability 
unless a large portion of state variables is secured. 
Moreover, similar behavior is witnessed for both the \ac{dc} and \ac{ac} models, where for both models, the proposed protection scheme 
significantly enhances the detection probability by protecting only a small percentage of state variables.

 \section{Conclusions} \label{sec; conclusions}
We introduce a new, \ac{gsp}-based, defensive approach against \ac{fdi} attacks in power systems. First, from the adversary’s point of view,  we present a new unobservable \ac{fdi} attack, the \ac{gfdi} attack, that utilizes the graph properties of the states to bypass the recently developed \ac{gsp}-based detectors.
Then, from the perspective of the system operator, we present countermeasures against the \ac{gfdi} attacks. The proposed protection scheme aims to select a minimal set of sensors to prevent the success of \ac{gfdi} attacks by forcing the attack to have a high graph TV and, thus, enabling its detection by advanced GSP tools.
This approach requires a smaller set of secured states than existing designs, which translates to a lower construction cost and a shorter installation time for new hardware. 
The proposed \ac{gfdi} attack design and protection scheme are applicable for both \ac{dc} and \ac{ac} models. Our numerical simulations show that existing detection methods have a significantly lower detection probability for the proposed \ac{gfdi} attack compared to previous attack designs, indicating the significant threat posed by the \ac{gfdi} attack to power systems. Moreover, the simulations demonstrate 
 that our proposed \ac{gsp}-based protection scheme requires a smaller set of secured sensors
 compared to existing designs,
  resulting in lower construction costs and shorter installation times. 
   Therefore, our approach provides a new and cost-effective solution for enhancing the resilience of power systems against FDI attacks.
Future research should include a practical investigation of the \ac{gsp}-based detectors, attacks, and protection schemes under various real-world settings.
 Other research directions involve extending the \ac{gfdi} attack design and the proposed protection scheme to partially observable systems, as well as optimizing the protection scheme.
In addition, the proposed approach, including the attack and protection design, could be extended to general sensor networks that are based on \ac{gsp} data tasks.
\begin{appendices}
\renewcommand{\thesectiondis}[2]{\Alph{section}:}

\vspace{-0.2cm}
\section{Proof for Theorem \ref{thm; gfdi cases}} \label{App; cases}

Let $\cvec^*$ and $\cvec^{**}$ be the optimal solutions of \eqref{eq; gfdi start} and \eqref{eq; gfdi cases}, respectively.
In the following, we  first show that  $(\cvec^{**})^T\Lmat\cvec^{**} \le (\cvec^{*})^T\Lmat\cvec^{*}$ 
and then that $\cvec^{**}$ is in the feasible set of \eqref{eq; gfdi start}. 
If both requirements are met, then $\cvec^{**}$ is a feasible solution in the minimization problem \eqref{eq; gfdi start} with a cost function smaller than the cost function 
associated with the optimal solution, $\cvec^*$.
Therefore, $\cvec^{**}$ is the optimal solution of \eqref{eq; gfdi start}, i.e., $\cvec^*=\cvec^{**}$. 

Without loss of generality, let the index $j$ be such that  $\tlvert\cvec^*\trvert_{\infty}=\lvert c^*_j \rvert $ and define $\bar{\cvec}^*={\text{sign}}(c^*_j)\frac{\lvert c_j^* \rvert}{\tau}\cvec^*$, 
 where ${\text{sign}}(\cdot)$ denotes the sign function, which assigns $1$ to positive argument and $-1$ for negative ones. 
As a result, we obtain that
\be \label{eq; tv appendix}
(\bar{\cvec}^*)^T\Lmat\bar{\cvec}^*=(\tau/\lvert c_j^*\rvert)^2 (\cvec^*)^T\Lmat\cvec^*\le(\cvec^*)^T\Lmat\cvec^*,
\ee
where $\tlvert\cvec^*\trvert_{\infty}=\lvert c^*_j \rvert \ge\tau$.
In addition, it can be observed that the definition of $\bar{\cvec}^*$ ensures  $\bar{c_j}^{*}=\tau$. 
Moreover, because $\cvec^*$ is a feasible solution in \eqref{eq; gfdi start}, 
and thus, satisfies the last constraint in \eqref{eq; gfdi start},
we obtain that
$$
\Hmat^{\mathcal{S}}\bar{\cvec}^*={\text{sign}}(c^*_j)\frac{\lvert c_j^* \rvert}{\tau}\Hmat^{\mathcal{S}} \cvec^*=\zerovec,
$$
and 
$$
\|\bar{\cvec}^*\|_0= \|{\text{sign}}(c^*_j)\frac{\lvert c_j^* \rvert}{\tau}\cvec^*\|_0 =\|\cvec^*\|_0\le k.
$$
Thus, $\bar{\cvec}^*$ is also a feasible solution for the inner minimization of \eqref{eq; gfdi cases}
for the case $i=j$.
 Consequently, $\bar{\cvec}^*$  is a feasible solution in \eqref{eq; gfdi cases}. 
As a result, its cost function is ensured to be higher than or equal to the cost function of the optimal result in 
\eqref{eq; gfdi cases}, i.e., $(\cvec^{**})^T\Lmat\cvec^{**} \le (\bar{\cvec}^*)^T\Lmat\bar{\cvec}^*$. 
Hence, from \eqref{eq; tv appendix} we obtain $(\cvec^{**})^T\Lmat\cvec^{**} \le (\cvec^*)^T\Lmat\cvec^*$.
Now, let $i$ be the case that minimizes the outer minimization in \eqref{eq; gfdi cases}, 
then the optimal solution of \eqref{eq; gfdi cases}, $\cvec^{**}$, satisfies $\lvert c^{**}_{i} \rvert =\tau$, $\tlvert \cvec^{**}\trvert_0\le k$, 
and $\Hmat^{\mathcal{S}}\cvec^{**}=\zerovec$. Hence, $\cvec^{**}$ is a feasible solution for \eqref{eq; gfdi start}.

\end{appendices}

\vspace{-0.3cm}
\section*{Acknowledgments}
  This work was supported in part by the Next Generation Internet (NGI) program, the Jabotinsky Scholarship from the Israel Ministry of Technology and Science, the Israel Ministry of National Infrastructure, Energy, National Research Foundation of Korea (NRF) grant funded by the Korean government (MSIT) (No. RS-2023-00210018), NSF grants CNS-2148128, EPCN-2144634, EPCN-2231350, and by the U.S. Department of Energy’s Office of Energy Efficiency and Renewable Energy under the Solar Energy Technology Office Award Number DE-EE0008769. The views expressed herein do not necessarily represent the views of the U.S. Department of Energy or the United States Government.
 The authors would like to thank Prof. Eran Treister from the Department of Computer Science at Ben-Gurion University for making a valuable contribution to the implementation of the GFDI attack.
The authors thank the anonymous reviewers for their constructive comments, which helped to improve the quality of the paper and the clarity of the proposed attack strategy.

\vspace{-0.3cm}

\bibliographystyle{IEEEtran}
\bibliography{FDI_GT}

\vspace{-0.3cm}
\vskip -2\baselineskip plus -1fil
\begin{IEEEbiographynophoto}
{Gal Morgenstern} 
received his B.Sc. (cum laude) and M.Sc. degrees in 2019 and 2020, respectively, from Ben-Gurion University of the Negev, Israel, in Electrical and Computer Engineering. He is currently in his last year as a Ph.D. research student at the School of Electrical and Computer Engineering at Ben-Gurion University of the Negev, Beer-Sheva, Israel. 
He was awarded the Israeli Ministry of Science and Technology Zabutinsky scholarship in 2019, and the Kreitman Negev scholarship in 2019,
and has been selected to participate in the UN-funded Next Generation Internet (NGI) Explorers Program in 2021. 
His main research interests include 
statistical signal processing and graph signal processing, with applications to power system cyber security. 
\end{IEEEbiographynophoto} 

\vspace{-0.3cm}
\vskip -2\baselineskip plus -1fil
\begin{IEEEbiographynophoto}
{Jip Kim} is an assistant professor in the Department of Energy Engineering at Korea Institute of Energy Technology (KENTECH). Prior to joining KENTECH, Jip worked as a postdoctoral research scientist at the Electrical Engineering department at Columbia University from 2021 to 2022. He received the Ph.D. degree in Electrical Engineering from the Smart Energy Research Laboratory at New York University (NYU), the B.S and M.S degrees in Electrical Engineering from Yonsei University and Seoul National University. His research focuses on developing mathematical models and optimization algorithms to solve power system engineering and energy economics problems.
\end{IEEEbiographynophoto} 

\vspace{-0.3cm}
\vskip -2\baselineskip plus -1fil
\begin{IEEEbiographynophoto}
{James Anderson} is an assistant professor in the Department of Electrical Engineering at Columbia University, where he is also a member of the Data Science Institute. From 2016 to 2019 he was a senior postdoctoral scholar in the Department of Computing + Mathematical Sciences  at the California Institute of Technology. Prior to Caltech, he held a Junior Research Fellowship at St John's College at the University Oxford and was also affiliated with the Department of Engineering Science. He was awarded a DPhil (PhD) from Oxford in 2012 and the BSc and MSc degrees from the University of Reading in 2005 and 2006 respectively. His research interests include distributed, robust, and optimal control and large-scale optimization with applications to smart power grids.
\end{IEEEbiographynophoto} 

\vspace{-0.3cm}
\vskip -2\baselineskip plus -1fil
\begin{IEEEbiographynophoto}
{Gil Zussman}  received the Ph.D. degree in electrical engineering from the Technion in 2004 and was a postdoctoral associate at MIT in 2004–2007. He has been with Columbia University since 2007, where he is a Professor of Electrical Engineering and Computer Science (affiliated faculty). His research interests are in the area of networking, and in particular in the areas of wireless, mobile, and resilient networks. He is an IEEE Fellow, received the Fulbright Fellowship, two Marie Curie Fellowships, the DTRA Young Investigator Award, and the NSF CAREER Award. He is a co-recipient of 7 paper awards including the ACM SIGMETRICS’06 Best Paper Award, the 2011 IEEE Communications Society Award for Advances in Communication, and the ACM CoNEXT’16 Best Paper Award. He has been the TPC chair of IEEE INFOCOM’23, ACM MobiHoc’15, and IFIP Performance 2011, and is the Columbia PI of the NSF PAWR COSMOS testbed. 
\end{IEEEbiographynophoto} 

\vspace{-0.3cm}
\vskip -2\baselineskip plus -1fil
\begin{IEEEbiographynophoto}
{Tirza Routtenberg} received the B.Sc. degree (magna cum laude)  in Biomedical Engineering from the Technion Israel Institute of Technology,  Haifa, Israel, in 2005, and the M.Sc. (magna cum laude) and Ph.D. degrees in Electrical Engineering from Ben-Gurion University of the Negev, Beer-Sheva, Israel, in 2007 and 2012, respectively. She was a Postdoctoral fellow with the School of Electrical and Computer Engineering, Cornell University, in 2012-2014. Since October 2014, she is a faculty member at the School of Electrical and Computer Engineering, Ben-Gurion University of the Negev, Beer-Sheva, Israel. In 2022–2023, she is a William R. Kenan, Jr. Visiting Professor for Distinguished Teaching at Princeton University.  Her research interests include signal processing in smart grids, statistical signal processing, and graph signal processing. She is an associate editor of IEEE Transactions on Signal and Information Processing Over Networks and of IEEE Signal Processing Letters.
She is a co-recipient of 4 Best Student Paper Awards at ICASSP 2011,  CAMSAP 2013,  ICASSP 2017, and IEEE Workshop on SSP 2018. She was awarded the Negev scholarship in 2008, the Lev-Zion scholarship in 2010, the Marc Rich Foundation Prize in 2011, and the Toronto Prize for Excellence in Research in 2021.
\end{IEEEbiographynophoto}

\end{document}